\documentclass[preprint,11pt]{elsarticle}
\usepackage{amsfonts}   
\usepackage{amssymb,amsmath,amsthm}    

\catcode`\@=11 

\newtheorem{theorem}{Theorem}[section]
\newtheorem{lemma}{Lemma}[section]

\newtheorem{corollary}{Corollary}[section]
\newtheorem{proposition}{Proposition}[section]

\newcommand{\prf}{\begin{proof}{Proof}}

\def\A{{\cal A}} \def\a{{\alpha}}
 \def\b{{\beta}}                         
                                        
 \def\d{{\delta}}      
                 \def\e{{\varepsilon}}                   
                          
 \def\g{{\gamma}}

 \def\l{{\lambda}}     
 \def\m{{\mu}}                          
                 \def\n{{\nu}}

 \def\q{{\theta}}                       
 \def\r{{\rho}}				\def\RR{\mathbb{R}}
                 \def\s{{\sigma}}

\def\X{{\cal X}} 				 
					 
                         \def\ZZ{\mathbb{Z}}

\def\sumin{{\sum_{i=1}^n}}

\def\sumik{{\sum_{i=1}^k}}
\def\sumjk{{\sum_{j=1}^k}}

\def\half{{\frac12}}
\def\thalf{{\textstyle\half}}

\def\ra{\rightarrow}
\def\da{\downarrow}

\def\dotfil{\leaders\hbox to 1em{\hss.\hss}\hfill}
\def\minqquad{\hskip-2em\relax}
\def\qqqquad{\hskip4em\relax}

\def\hexnumber#1{\ifcase#1 0\or 1\or 2\or 3\or 4\or 5\or 6\or 7\or 8\or 9\or A\or B\or C\or D\or E\or F\fi}

\def\BL{{\rm B\kern -0.5pt L}}
\def\UC{{\rm U\kern-0.5pt C}}

\def\lin{\mathop{\rm lin\,}\nolimits}

\def\tdot#1{\kern1.2pt\dot{\vphantom{#1}}\kern-1.2pt
        \dot#1\kern0.8pt\dot{\vphantom{#1}}\kern-0.8pt}

\def\E{\mathord{\rm E}}
\def\Pr{\mathord{\rm P}}
\def\var{\mathop{\rm var}\nolimits}

\mathchardef\given="626A
\mathcode`:="603A

\def\generalweak#1{\ {\mathchoice{\buildrel #1\over \rightsquigarrow}%
{\raise-2pt\hbox{$\buildrel #1\over \rightsquigarrow$}}{}{}}\ }

\def\generalprob#1{\ {\mathchoice{\raise-1.5pt\hbox{$\buildrel#1\over\ra$}}
{\raise-2pt\hbox{$\buildrel#1\over\ra$}}{}{}}\ }

\def\prob{\generalprob{\small P}}

\def\Pgg{\ {\mathchoice{\raise-1.5pt\hbox{$\buildrel {\small P}\over\gg$}}
{\raise-2pt\hbox{$\buildrel {\small P}\over\gg$}}{}{}}\ }

\def\beginskip{\begingroup\long\def\eet##1\endskip{\endgroup}\eet}
\def\endskip{}
\def\boxit#1{\vbox{\hrule\hbox{\vrule\kern1pt\vbox{\kern1pt #1\kern1pt}\kern1pt\vrule}\hrule}}
\def\eet#1{}

\catcode`\@=12

\newcommand{\LEB}{\l}

\def\dsum{\mathop{\mathchar"1350\mathchar"1350}\limits}
\def\dint{\mathop{\int\!\!\int}\limits}
\let\sd\sdev

\journal{Stochastic Processes and its Applications}

\begin{document}

\begin{frontmatter}



\title{Asymptotic Normality of Quadratic Estimators}


\address{}

\author{James Robins}
\ead{robins@hsph.harvard.edu}
\author{Lingling Li}\ead{lili@hsph.harvard.edu}
\author{Eric Tchetgen Tchetgen}\ead{etchetge@hsph.harvard.edu}
\author{Aad van der Vaart\fnref{erc}}\ead{avdvaart@math.leidenuniv.nl}
\fntext[erc]{The research leading to these results has 
received funding from the European Research Council 
under\  ERC Grant Agreement 320637.}

\address{Departments of Biostatistics and Epidemiology\\
School of Public Health\\
Harvard University\fnref{label3}}
\address{Mathematical Institute\\
Leiden University}

\begin{abstract}
We prove conditional asymptotic normality of a class of 
quadratic U-statistics that are
dominated by their degenerate second order part and
have kernels that change with the number of observations. 
These statistics arise in the construction of estimators
in high-dimensional semi- and non-parametric models,
and in the construction of nonparametric confidence sets.
This is illustrated by estimation of the
integral of a square of a density or regression function,
and estimation of the mean response with missing data.
We show that estimators are asymptotically normal
even in the case that the rate is slower than the square
root of the observations.
\end{abstract}


\end{frontmatter}

\section{Introduction}
Let $(X_1,Y_1),\ldots, (X_n,Y_n)$ be i.i.d.\ random vectors, taking
values in sets $\X\times\RR$, for an arbitrary 
measurable space $(\X,\A)$ and $\RR$ equipped with the
Borel sets. For given symmetric, measurable functions 
$K_n: \X\times\X\to \RR$ consider the $U$-statistics
\begin{equation}
U_n={1\over n(n-1)}\dsum_{1\le r\not=s\le n}K_n(X_r,X_s)Y_rY_s.
\label{EqU}
\end{equation}
Would the kernel $(x_1,y_1,x_2,y_2)\mapsto K_n(x_1,x_2)y_1y_2$ 
of the $U$-statistic be independent of $n$ and have 
a finite second moment,
then either the sequence $\sqrt n(U_n-\E U_n)$ would be asymptotically normal
or the sequence $n (U_n-\E U_n)$ would converge in distribution
to Gaussian chaos. The two cases can be described in terms
of the Hoeffding decomposition $U_n=\E U_n+U_n^{(1)}+U_n^{(2)}$
of $U_n$, where $U_n^{(1)}$ is the best approximation of $U_n-\E U_n$ by
a sum of the type $\sumin h(X_r,Y_r)$ and $U_n^{(2)}$ is the remainder,
a degenerate $U$-statistic (compare (\ref{EqHoeffding})
in Section~\ref{SectionProofs}). 
For a fixed kernel $K_n$ the linear term $U_n^{(1)}$ dominates
as soon as it is nonzero, in which case  asymptotic normality pertains;
in the other case $U_n^{(1)}=0$ and 
the $U$-statistic possesses a nonnormal limit distribution.

If the kernel depends on $n$, then the separation between the linear and
quadratic cases blurs.  In this paper we are interested in this
situation and specifically in kernels $K_n$ that concentrate as
$n\ra\infty$ more and more near the diagonal of $\X\times \X$.  In our
situation the variance of the $U$-statistics is dominated by the
quadratic term $U_n^{(2)}$.  However, we show that the sequence
$(U_n-\E U_n)/\s(U_n)$ is typically still asymptotically normal.
The intuitive explanation is that the $U$-statistics
behave asymptotically as ``sums across the diagonal $r=s$'' and thus
behave as sums of independent variables. Our formal proof is based on
establishing conditional asymptotic normality given a binning of the
variables $X_r$ in a partition of the set $\X$.

Statistics of the type (\ref{EqU}) arise in many problems of estimating a
functional on a semiparametric model, with $K_n$ the kernel of a
projection operator  (see \cite{FreedmanFestschrift}). As
illustrations we consider in this paper the problems of estimating
$\int g^2(x)\,dx$ or $\int f^2(x)\,dG(x)$, where $g$ is a density and
$f$ a regression function, and of estimating the mean treatment effect in
missing data models.  Rate-optimal estimators in the first of these three problems
were considered by \cite{BickelRitov88, BirgeMassart95, Laurent96, Laurent97, LaurentMassart00}, among others.
In Section~\ref{SectionExamples} we prove
 asymptotic normality of the estimators in \cite{Laurent96,Laurent97}, also in the
case that the rate of convergence is slower than $\sqrt n$, usually
considered to be the ``nonnormal domain''.  For the second and third problems
estimators of the form (\ref{EqU}) were derived 
in \cite{FreedmanFestschrift,RobinsMetrika,vdVStatSci,RobinsAOS}
using the theory of second-order estimating equations. Again we show that these are
asymptotically normal, also in the case that the rate is slower than $\sqrt n$.

Statistics of the type (\ref{EqU}) also arise in the construction of 
adaptive confidence sets, as in \cite{RobinsvdV06}, where
the asymptotic normality can be used to set precise confidence limits.

Previous work on $U$-statistics with kernels that depend on $n$ 
includes \cite{Weber,BhattacharyaGhosh,JammaladakaRaoJanson,deJong87,deJong90}.
These authors prove unconditional asymptotic normality using
the martingale central limit theorem, under somewhat different conditions. Our proof uses
a Lyapounov central limit theorem (with moment $2+\e$) combined with a conditioning argument,
and an inequality for moments of $U$-statistics due to E.\ Gin\'e.
Our conditions relate directly to the contraction of the kernel, and can be verified 
for a variety of kernels. The conditional form of our limit result
should be useful to separate different roles for the observations, such
as for constructing preliminary estimators and for constructing estimators of functionals.
Another line of research (as in as in \cite{Mikosch}) is concerned
with $U$-statistics that are well approximated by their projection
on the initial part of the eigenfunction expansion.
This has no relation to the present work, as here
the kernels explode and the $U$-statistic is 
asymptotically determined by the (eigen) directions ``added''
to the kernel as the number of observations increases.
By making special choices of kernel and variables $Y_i$,
the statistics (\ref{EqU}) can reduce to certain chisquare statistics,
studied in \cite{Morris, Ermakov}.

The paper is organized as follows.  In Section~\ref{SectionMainResult} we
state the main result of the paper, the asymptotic normality of
$U$-statistics of the type (\ref{EqU}) under general conditions on the
kernels $K_n$. Statistical applications are given in Section~\ref{SectionExamples}.
In Section~\ref{SectionProjectionKernels} 
the conditions of the main theorem shown to be satisfied by a variety of popular kernels,
including wavelet, spline, convolution, and Fourier kernels.
The proof of the main result is given in Section~\ref{SectionProofs}, while proofs for 
Section~\ref{SectionProjectionKernels} are given in an appendix.

The notation $a\lesssim b$ means $a\le C b$ for a 
constant $C$ that is fixed in the context. The notations
$a_n\sim b_n$ and $a_n\ll b_n$ mean that $a_n/b_n\ra 1$ and
$a_n/b_n\ra 0$, as $n\ra\infty$. 
The space $L_2(G)$ is the set of
measurable functions $f: \X\to \RR$ that are square-integrable relative
to the measure $G$ and $\|f\|_G$ is the corresponding norm.
The product $f\times g$ of two functions is to be understood
as the function $(x_1,x_2)\mapsto f(x_1)g(x_2)$, whereas the product
$F\times G$ of two measures is the product measure.

\section{Main result}
\label{SectionMainResult}
In this section we state the main result of the paper, 
the asymptotic normality of the $U$-statistics (\ref{EqU}), under
general conditions on the kernels $K_n$ and distributions of 
the vectors $(X_r,Y_r)$. For $q>0$ let
\begin{align*}
\m(x)&=\E \bigl(Y_1\given X_1=x\bigr),\\
\m_q(x)&=\E \bigl(|Y_1|^q\given X_1=x\bigr)
\end{align*}
be versions of the conditional (absolute) moments of $Y_1$ given $X_1$.
For simplicity we assume that $\m_1$ and and $\m_2$ are uniformly bounded.
The marginal distribution of $X_1$ is denoted by $G$.

The kernels are assumed to be 
measurable maps $K_n: \X\times\X\to \RR$ that are symmetric
in their two arguments and satisfy 
$\int\!\!\int K_n^2\,d(G\times G)<\infty$ for
every $n$. Thus the corresponding  {\sl kernel operators}
(with abuse of notation denoted by the same symbol)
\begin{equation}
K_nf(x)=\int f(v)K_n(x,v)\,dG(v),
\label{EqKernelOperator}
\end{equation}
are continuous, linear operators
$K_n: L_2(G)\to L_2(G)$. We assume that their operator norms $\|K_n\|
=\sup\{\|K_nf\|_G: \|f\|_G=1\}$ are uniformly bounded:
\begin{equation}
\sup_n \|K_n\|<\infty.
\label{EqBoundedNorms}
\end{equation}
By the Banach-Steinhaus theorem this is certainly the case if $K_n f\ra f$ in $L_2(G)$
as $n\ra\infty$ for every $f\in L_2(G)$.
The operator norms $\|K_n\|$ are
typically much smaller than the $L_2(G\times G)$-norms of the kernels.
The squares of the latter are typically of the same order of magnitude
as the square  $L_2(G\times G)$-norms weighted by $\m_2\times\m_2$, 
which we denote by
\begin{equation}
k_n:= \int\!\!\int K_n^2(x,y)\ (\m_2\times\m_2)(x,y)\,d(G\times G)(x,y).
\label{EqDefinitionkn}
\end{equation}
We consider the situation that these square weighted norms are 
strictly larger than $n$:
\begin{equation}
{k_n\over n}\ra \infty.
\label{EqKnToInfinity}
\end{equation}
Under condition (\ref{EqKnToInfinity}) the variance of the $U$-statistic
(\ref{EqU}) is dominated by the variance of the quadratic part of
its Hoeffding decomposition.
In contrast, if $k_n=n$, the linear and quadratic parts contribute
variances of equal order. This case can be handled by the
methods of this paper, but requires  a special discussion, which
we omit. The remaining case $k_n\ll n$ 
leads to asymptotically linear $U$-statistics,
and is well understood.

The remaining conditions concern the concentration of the kernels $K_n$  to the diagonal
of $\X\times\X$. We assume that there exists a sequence of finite partitions
$\X=\cup_m\X_{n,m}$ in measurable sets such that 
\begin{align}
{1\over k_n}\sum_m\int_{\X_{n,m}}\!\!\int_{\X_{n,m}} K_n^2\
(\m_2\times\m_2)\,d(G\times G)&\ra 1,
\label{EqConditionOne}\\
{1\over k_n}\max_m\int_{\X_{n,m}}\!\!\int_{\X_{n,m}} K_n^2\ 
(\m_2\times\m_2)\,d(G\times G)&\ra 0,
\label{EqConditionTwo}\\
\max_m G(\X_{n,m})&\ra0,
\label{EqConditionTwoHalf}\\
\liminf_{n\ra\infty} n\, \min_m G(\X_{n,m})&>0.
\label{EqConditionThree}
\end{align}
The sum in the first condition (\ref{EqConditionOne}) is  the integral of the square kernel 
(weighted by the function $\m_2\times\m_2$) over the set $\cup_m(\X_{n,m}\times\X_{n,m})$
(shown in Figure~\ref{FigurePartition}). The condition requires this to be
asymptotically equivalent to the integral $k_n$ of this same function
over the whole product space $\X\times \X$. 
The other conditions implicitly require that the partitioning sets
are not too different and not too numerous. 

\begin{figure}
\centerline{\resizebox{1.6in}{!}{\includegraphics{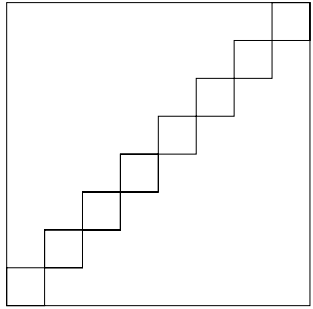}}}
\caption{The diagonal of $\X\times\X$ covered by the set 
$\cup_m(\X_{n,m}\times\X_{n,m})$.}
\label{FigurePartition}
\end{figure}

A final condition requires implicitly that the partitioning is fine enough. 
For some $q>2$, the partitions should satisfy
\begin{equation}
\frac1{k_n^{q/2}}\max_m\Bigl({G(\X_{n,m})\over n}\Bigr)^{q/2-1}
\sum_m\int_{\X_{n,m}}\!\!\int_{\X_{n,m}}\!
\!|K_n|^q (\m_q\times\m_q)\,d(G\times G)\!\ra 0.
\label{EqConditionFour}
\end{equation}
This condition will typically force the number of partitioning
sets to infinity at a rate depending on $n$ and $k_n$
(see Section~\ref{SectionProjectionKernels}). In the proof
it serves as a Lyapounov condition to enforce normality.

The existence of partitions satisfying the preceding conditions depends
mostly on the kernels $K_n$, and is established for various
kernels in Section~\ref{SectionProjectionKernels}. The following theorem
is the main result of the paper. Its proof is deferred to
Section~\ref{SectionProofs}.

Let $I_n$ be the vector with as coordinates $I_{n,1},\ldots, I_{n,n}$ 
the indices of the partitioning sets containing $X_1,\ldots, X_n$, i.e.\
$I_{n,r}=m$ if $X_r\in \X_{n,m}$. Recall that the bounded Lipschitz
distance generates the weak topology on probability measures.

\begin{theorem}
\label{TheoremMainResult}
Assume that the function $\m_2$ is uniformly bounded.
If (\ref{EqBoundedNorms}) and (\ref{EqKnToInfinity}) hold and there exist
finite partitions $\X=\cup_m\X_{n,m}$ such that 
(\ref{EqConditionOne})--(\ref{EqConditionFour}) hold,
then the bounded Lipschitz distance between the conditional law
of $(U_n-\E U_n)/\s(U_n)$ given $I_n$ and the standard normal distribution
tends to zero in probability.  Furthermore $\var U_n\sim 2k_n/n^2$
for $k_n$ given in (\ref{EqDefinitionkn}).
\end{theorem}

The conditional convergence in distribution implies the unconditional convergence. It expresses that
the randomness in $U_n$ is asymptotically determined by the fine positions of the $X_i$ within the
partitioning sets, the numbers of observations falling in the sets being fixed by $I_n$. 


In most of our examples the kernels are pointwise bounded above by
a multiple of $k_n$, and (\ref{EqDefinitionkn}) arises, because
the area where $K_n$ is significantly different from zero is of the order
$k_n^{-1}$. Condition (\ref{EqConditionFour}) can then be simplified to
\begin{equation}
\max_m G(\X_{n,m})\frac {k_n}n\ra 0.
\label{EqConditionFourSimplified}
\end{equation}

\begin{lemma}
\label{LemmaConditionFourSimplified}
Assume that the functions $\m_2$ and $\m_q$ are bounded away from zero
and infinity, respectively. If $\|K_n\|_\infty\lesssim k_n$,
then (\ref{EqConditionFour}) is implied by 
(\ref{EqConditionFourSimplified}).
\end{lemma}

\begin{proof}
The sum in (\ref{EqConditionFour}) is bounded up to a constant by
$\int\int |K_n|^q\,d(G\times G)$, which is bounded above
by a constant times $k_n^{q-2}\int\int K_n^2\,d(G\times G)\lesssim k_n^{q-1}$,
by the definition of $k_n$.
\end{proof}

\section{Statistical applications}
\label{SectionExamples}
In this section we give examples of statistical problems
in which statistics of the type (\ref{EqU}) arise as estimators.

\subsection{Estimating the integral of the square of a density}
\label{SectionDensityEstimation}
Let $X_1,\ldots, X_n$ be i.i.d.\ random variables with a density $g$
relative to a given measure $\n$ on a measurable space $(\X,\A)$. The problem of
estimating the functional  $\int g^2\,d\n$ has been addressed
by many authors, including \cite{BickelRitov88},
\cite{LaurentMassart00} and \cite{KerkPicard96}.
The estimators proposed by \cite{Laurent96,Laurent97}, which are particularly elegant, are based on an expansion of
$g$ on an orthonormal basis $e_1,e_2,\ldots$ of
the space $L_2(\X,\A,\n)$, so that
$\int g^2\,d\n=\sum_{i=1}^\infty \q_i^2$,
for $\q_i=\int ge_i\,d\n$ the Fourier coefficients of $g$.
Because $\E  e_i(X_1)e_i(X_2)=\q_i^2$, the square Fourier coefficient
$\q_i^2$ can be estimated unbiasedly by the $U$-statistic with kernel
$(x_1,x_2)\mapsto e_i(x_1)e_i(x_2)$.
Hence the truncated sum of squares $\sumik \q_i^2$ can be estimated
unbiasedly by
$$U_n=\sumik{1\over n(n-1)}\dsum_{r\not=s}e_i(X_r)e_i(X_s).$$
This statistic is of the type  (\ref{EqU}) with kernel
$K_n(x_1,x_2)=\sumik e_i(x_1)e_i(x_2)$ and the variables $Y_1,\ldots, Y_n$
taken equal to unity.

The estimator $U_n$ is unbiased for the truncated series $\sumik \q_i^2$, but
biased for the functional of interest $\int g^2\,d\n=\sum_{i=1}^\infty\q_i^2$.
The variance of the estimator can be computed to be of the order 
$k/n^2\vee1/n$ (cf.\ (\ref{EqVarU}) below). If the Fourier coefficients are known to 
satisfy $\sum_{i=1}^\infty \q_i^2i^{2\b}\le 1$, then the bias
can be bounded by $\sum_{i=k+1}^\infty\q_i^2
\le k^{-2\b}$, and trading square bias versus the variance leads to the choice  $k=n^{1/(2\b+1/2)}$.

In the case that $\b>1/4$, the mean square error of the estimator is
$1/n$ and the sequence $\sqrt n (U_n-\int g^2\,d\n)$ can be shown to
be asymptotically linear in the efficient influence function $2(g-\int g^2\,d\n)$  (see
(\ref{EqHoeffding}) with $\m(x)=\E(Y_1\given X_1=x)\equiv 1$ and 
\cite{Laurent96}, \cite{Laurent97}). 
More interesting from our present perspective is the
case that $0<\b< 1/4$, when the mean square
error is of order $n^{-4\b/(2\b+1/2)}\gg 1/n$,
and the variance of  $U_n$ is dominated by its second-order term. By
Theorem~\ref{TheoremMainResult} the estimator, centered at its
expectation, and with the orthonormal basis $(e_i)$ one of
the bases discussed in
Section~\ref{SectionProjectionKernels}, is still asymptotically normally
distributed.

\subsection{Estimating the integral of the square of a regression function}
\label{SectionRegression}
Let $(X_1,Y_1),\ldots, (X_n,Y_n)$ be i.i.d.\ random vectors
following the regression model $Y_i=b(X_i)+\e_i$ for unobservable
errors $\e_i$ that satisfy $\E(\e_i\given X_i)=0$. 
It is desired to estimate $\int b^2\,dG$ for $G$ the marginal
distribution of $X_1,\ldots, X_n$. 

If the distribution $G$ is
known, then an appropriate estimator can take
exactly the form (\ref{EqU}), for $K_n$ the kernel of an orthonormal projection 
on a suitable $k_n$-dimensional space in $L_2(G)$.
Its asymptotics are as in Section~\ref{SectionDensityEstimation}.

Because an orthogonal projection in $L_2(G)$ 
can only be constructed if $G$ is known, the preceding estimator
is not available if $G$ is unknown. If the regression function $b$ is regular
of order $\b\ge 1/4$, then the parameter can be estimated
at $\sqrt n$-rate (see \cite{FreedmanFestschrift}).
In this section we consider an estimator that is appropriate if 
$b$ is regular of order $\b<1/4$ and the design distribution
$G$ permits a Lebesgue density $g$ that is bounded away from zero and sufficiently smooth.

Given initial estimators $\hat b_n$ and $\hat g_n$ for the regression
function $b$ and design density $g$, we consider the estimator
\begin{align}
T_n&={1\over n} \sum_{r=1}^n \Bigl(\hat b_n(X_r)^2
  +2\hat b_n(X_r)\bigl(Y_r-\hat b_n(X_r)\bigr)\Bigr)\nonumber\\
&\qquad  +{1\over n(n-1)}\dsum_{1\le r\not=s\le n} \bigl(Y_r-\hat
  b_n(X_r)\bigr)K_{k_n,\hat g_n}(X_r,X_s) \bigl(Y_s-\hat   b_n(X_s)\bigr).
\label{EqRegressionEstimator}
\end{align} 
Here $(x_1,x_2)\mapsto K_{k,g}(x_1,x_2)$ is a projection kernel in the
space $L_2(G)$. For definiteness we construct this in the form
(\ref{EqProjectionKernelg}), where the basis $e_1,\ldots,e_k$ may be the
Haar basis, or a general wavelet basis, as discussed in
Section~\ref{SectionProjectionKernels}.  Alternatively, we could use
projections on the Fourier or spline basis, or convolution kernels, but
the latter two require twicing (see (\ref{EqTwicingKernel})) to control bias, and the arguments 
given below must be adapted.


The initial estimators $\hat b_n$ and $\hat g_n$ may be
fairly arbitrary rate-optimal estimators if constructed
from an independent sample of observations.  (e.g.\ after
splitting the original sample in parts used to construct
the initial estimators and the estimator (\ref{EqRegressionEstimator})). 
We assume this in the following theorem, and
also assume that the norm of $\hat b_n$ in $C^\b[0,1]$ is
bounded in probability, or alternatively, if the projection is on
the Haar basis, that this estimator is in the linear
span of $e_1,\ldots, e_{k_n}$. This is typically not
a loss of generality.

Let $\hat{\E}$ and $\hat{\var}$ denote expectation and variance
given the additional observations.
Set $\m_q(x)=\E (|\e_1|^q\given X_1=x)$ and let 
$\|\cdot\|_3$ denote the $L_3$-norm relative to Lebesgue measure.

\begin{corollary}
\label{CorollaryEstimatorIntegral}
Let $\hat b_n$ and $\hat g_n$ be estimators based on
independent observations that converge to $b$ and $g$
in probability relative to the uniform norm and satisfy
$\|\hat b_n-b\|_3=O_P(n^{-\b/(2\b+1)})$ and
$\|\hat g_n-g\|_3=O_P(n^{-\g/(2\g+1)})$.
Let $\m_q$ be  finite and uniformly bounded for some $q>2$.
Then for $b\in C^\b[0,1]$ and strictly positive $g\in C^\g[0,1]$,
with $\g\ge\b$, 
and  for $k_n$ satisfying (\ref{EqKnToInfinity}),
\begin{align*}
\Bigl|\hat{\E}_{b,g} T_n-\int b^2 dG\Bigr|
&=O_P\Bigl({1\over k_n}\Bigr)^{2\b}+
O_P\Bigl({1\over n}\Bigr)^{2\b/(2\b+1)+\g/(2\g+1)},\\
\hat{\var}_{b,g} T_n&= {2\over n^2}\int\!\!\int(\m_2\times \m_2)K_{k_n,g}^2\,
d(G\times G)\bigl(1+o_P(1)\bigr)=O_P\Bigl({k_n\over n^2}\Bigr).
\end{align*} 
Furthermore,
the sequence $(T_n-\hat{\E}_{b,g}T_n)/\hat{\sd}_{b,g} (T_n)$ tends
in distribution to the standard normal distribution.
\end{corollary}

For $k_n=n^{1/(2\b+1/2)}$ the estimator $T_n$ of $\int b^2\,dG$
attains a rate of convergence  of the order $n^{-2\b/(2\b+1/2)}+n^{-2\b/(2\b+1)-\g/(2\g+1)}$.
If $\g> \b/(4\b^2+\b+1/2)$, then this reduces to $n^{-4\b/(1+4\b)}$, which is known to be
the minimax rate when $g$ is known and $b$
ranges over a ball in $C^\b[0,1]$, for $\b\le 1/4$ (see \cite{BirgeMassart95} or \cite{RobinsMinimax}).
For smaller values of $\g$ the estimator 
can be improved by considering third or higher order
$U$-statistics (see \cite{RobinsAOS}).

\beginskip
It is possible to avoid using additional observations
and construct the initial estimators ``within the sample''
at the cost of some complications. 
Given a sequence of finite partitions $\X=\cup_m\X_{n,m}$
into an even number of sets,
we split the sample of observations $(X_1,Y_1),\ldots ,(X_n,Y_n)$
into two sets, depending on whether a variable $X_r$ falls in 
an even-indexed partitioning set
(i.e.\ $X_r\in\X_1:=\cup_m \X_{n,2m}$) or in an odd-indexed set 
(i.e.\ $X_r\in\X_2:=\cup_m \X_{n,2m+1}$).
Next we construct two estimators, $T_{n,1}$ and
$T_{n,2}$, in the same way as $T_n$, but based on the respective
sets of observations and with the initial estimators for
$b$ and $g$ determined by the other set of observations.
The final estimator is $T_n=T_{n,1}+T_{n,2}$.

If we choose the partition $\X=\cup_m\X_{n,m}$ a
coarsening of the partition implicit in the Haar basis at dimension $k_n$,
which determines the Haar kernel $K_n$, then the terms $(r,s)$ in the
double sum in the definition of $T_n$ with $X_r$ and $X_s$ not in
the same group vanish. Thus the splitting of the observations
in two groups then corresponds to a splitting of 
$T_n$ into two parts, given by
\begin{align*}
&\qquad T_{n,i}={1\over n} \sum_{r=1}^n \Bigl(\hat b_{n,i}(X_r)^2
+2\hat b_{n,i}(X_r)\bigl(Y_r-\hat b_{n,i}(X_r)\bigr)\Bigr)1_{\X_i}(X_r)\\
&\!+{1\over n(n-1)}\dsum_{1\le r\not=s\le n}
\bigl(Y_r-\hat b_{n,i}(X_r)\bigr)K_{k_n,\hat g_{n,i}}(X_r,X_s)
1_{\X_i\times\X_i}(X_r,X_s)
\bigl(Y_s-\hat b_{n,i}(X_s)\bigr).
\end{align*}
The estimators $\hat b_{n,i}$ are constructed from the
observations $(X_r,Y_r)$ with $X_r\not\in \X_i$, for $i=1,2$.
Because the Haar partitioning of dimension $k_n$ is much
finer than the resolution level needed to estimate $b$ and $g$
(for which a partition in $n^{1/(2\b+1)}$ sets suffices), the
rates of the estimators $\hat\b_{n,i}$ and $\g_{n,i}$ need not suffer
from splitting the sample according to the sets $\X_i$.

The estimators $T_{n,i}$ have the same structure
as the estimator $T_n$,  and can be studied by the same methods.

\endskip

\subsection{Estimating the mean response with missing data}
Suppose that a typical observation is distributed as
$X=(YA,A,Z)$ for $Y$ and $A$ taking values in the
two-point set $\{0,1\}$ and conditionally independent
given $Z$, with conditional mean functions 
$b(z)=\Pr(Y=1\given Z=z)$ and $a(z)^{-1}=\Pr(A=1\given Z=z)$, and 
$Z$ possessing density $g$ relative to some dominated measure $\n$. 

In \cite{RobinsMetrika} we introduced a quadratic estimator for
the \emph{mean response} $\E Y=\int bg\,d\n$, which attains a better
rate of convergence than the conventional linear estimators.
For initial estimators $\hat a_n$, $\hat b_n$ and $\hat g_n$,
and $K_{k,\hat\a_n,\hat g_n}$ a projection kernel in $L_2(g/a)$, this takes the form
\begin{align*}
&\frac1n\sum_{r=1}^n \Bigl(A_r \hat a_n(Z_r)\bigl(Y_r-\hat b_n(Z_r)\bigr)+\hat b_n(Z_r\Bigr)\\
&\quad-\frac1{n(n-1)}\dsum_{1\le r\not=s\le n}
\Bigl(A_r\bigl(Y_r-\hat b_n(Z_r)\bigr)K_{k_n,\hat\a_n,\hat g_n}(Z_r,Z_s)
\bigl(A_s\hat a_n(Z_s)-1\bigr)\Bigr).
\end{align*}
Apart from the (inessential) asymmetry of the kernel, the quadratic part has the form
(\ref{EqU}). Just as in the preceding section, the estimator can be shown to be asymptotically
normal with the help of Theorem~\ref{TheoremMainResult}.

\section{Kernels}
\label{SectionProjectionKernels}
In this section we discuss examples of kernels that satisfy the conditions of our main
result. Detailed proofs are given in an appendix. 

Most of the examples are kernels of \emph{projections} $K$, 
which are characterised by the identity $Kf=f$, for every $f$ in their range
space. For a projection given by a kernel, the latter is equivalent to $f(x)=\int f(v) K(x,v)\,dG(v)$ for
(almost) every $x$, which suggests that the measure $v\mapsto K(x,v)\,dG(v)$
acts on $f$ as a Dirac kernel located at $x$. Intuitively,
if the projection spaces increase to the full space, so that the identity is true for more
and more $f$, then the kernels $(x,v)\mapsto K(x,v)$ must be increasingly 
dominated by their values near the diagonal, thus 
meeting the main condition of Theorem~\ref{TheoremMainResult}.

For a given orthonormal basis $e_1,e_2,\ldots$ of $L_2(G)$,
the orthogonal projection onto $\lin(e_1,\ldots, e_k)$ is 
the kernel operator $K_k: L_2(G)\to L_2(G)$ with kernel
\begin{equation}
K_k(x_1,x_2)=\sumik e_i(x_1)e_i(x_2).
\label{EqProjectionKernelOnBasis}
\end{equation}
It can be checked that it has operator norm 1, while the
square $L_2$-norm $\int\int K_k^2\,d(G\times G)=k$ of the kernel is $k$.

\beginskip
In applications we may use a projection kernel corresponding
to a measure $G$ that is not the distribution of the 
variables $X_r$. If the latter measure has a well-behaved
density relative to $G$, this does not matter for the validity
of the main result. The following lemma is obvious, but useful.

\begin{lemma}
\label{LemmaProductKernelAndFunction}
If the conditions (\ref{EqBoundedNorms}), (\ref{EqConditionOne})--(\ref{EqConditionFour})
of Theorem~\ref{TheoremMainResult} hold for
kernels $K_n$, then they also hold for the kernels
$(x_1,x_2)\mapsto K_n(x_1,x_2)f_n(x_1,x_2)$ for measurable functions $f_n$
that are uniformly bounded above and below by positive
constants.
\end{lemma}
\endskip

A given orthonormal basis $e_1,e_2,\ldots$ relative to a given dominating  measure, 
can be turned into an orthonormal basis $e_1/\sqrt g,e_2/\sqrt g,\ldots$
of $L_2(G)$, for $g$ a density of $G$. The kernel of the orthogonal projection 
in $L_2(G)$ onto $\lin (e_1/\sqrt g,\ldots,e_k/\sqrt g)$ is
\begin{equation}
K_{k,g}(x_1,x_2)={\sumik e_i(x_1)e_i(x_2)\over \sqrt {g(x_1)}\sqrt {g(x_2)}}.
\label{EqProjectionKernelg}\end{equation}
If $g$ is bounded away from zero and infinity, the conditions
of Theorem~\ref{TheoremMainResult} will hold for this kernel
as soon as they hold for the kernel 
(\ref{EqProjectionKernelOnBasis}) relative to the dominating measure.

The orthogonal projection in $L_2(G)$ onto
the linear span $\lin(f_1,\ldots, f_k)$ of an arbitrary
set of functions $f_i$ possesses the kernel
\begin{equation}
K_k(x_1,x_2)=\sumik\sumjk A_{i,j} f_i(x_1)f_j(x_2),
\label{EqProjectionOnFixedSpace}\end{equation}
for $A$ the inverse of the $(k\times k)$-matrix with $(i,j)$-element
$\langle f_i,f_j\rangle_G$. In statistical applications
this projection has the advantage that it projects
onto a space that does not depend on the (unknown) measure $G$.
For the verification of the conditions of Theorem~\ref{TheoremMainResult} it is useful
to note that the matrix $A$ is well-behaved if 
$f_1,\ldots,f_k$ are orthonormal relative to a
measure $G_0$ that is not too different from $G$: from the identity
$\a^T \bigl(\langle f_i,f_j\rangle_G\bigr)\a=\int(\sumik \a_i f_i)^2\,dG$,
one can verify that the eigenvalues of $A$ are bounded away from zero and
infinity if $G$ and $G_0$ are absolutely continuous with a density
that is bounded away from zero and infinity.

\beginskip
\begin{lemma}
\label{LemmaProjectionOnFixedSpace}
If $f_1,f_2,\ldots$ is a an orthonormal basis for a
measure relative to which $G$ has a density that is bounded 
away from 0 and infinity, then the 
eigenvalues of the matrix $A$ in (\ref{EqProjectionOnFixedSpace}) 
are bounded away from zero and infinity.
\end{lemma}

\begin{proof}
For any vector $\a\in \RR^k$, we have that 
$\a^T \bigl(\langle f_i,f_j\rangle_G\bigr)\a=\int(\sumik \a_i f_i)^2\,dG$.
The same equation is valid with $G$ replaced by the
dominating measure. By the assumption on $G$ 
the quotient of the right sides, 
and hence also of the left sides, of the equations
is bounded away from zero and infinity. This shows that the eigenvalues
of the matrix $A^{-1}$ are bounded away from zero and infinity,
which is then also true for the matrix $A$. 
\end{proof}

Orthogonal projections arise as best approximations
by a function in a given finite-dimensional space in $L_2(G)$.
In statistical applications we may need good approximation
properties for a variety of norms, including for instance
the uniform norm, and good numerical properties. Nonorthogonal
projections, for instance on a spline basis, or other approximations,
such as convolution kernels,
can therefore also be of interest. 
\endskip

Orthogonal projections $K$ have the important property of making the inner product
$\langle (I-K) f, f\rangle_G=\|(I-K)f\|_G^2$ 
quadratic in the approximation error. Nonorthogonal
projections, such as the convolution kernels or spline kernels discussed below, lack this property, and may result in 
a large bias of an estimator. {\sl Twicing kernels},
discussed in \cite{Neweyetal} as a means to control
the bias of plug-in estimators, remedy this problem.
The idea is to use the operator $K+K^*-KK^*$, where
$K^*$ is the adjoint of $K: L_2(G)\to L_2(G)$, 
instead of the original operator $K$.
Because $I-K-K^*+KK^*=(I-K)(I-K^*)$, it follows that
$$\bigl\langle (I-K-K^*+KK^*)f, f\bigr\rangle_G
=\bigr\langle (I-K)f, (I-K)f\bigr\rangle_G=
\bigr\|(I-K)f\bigr\|^2_G.$$
If $K$ is an orthogonal projection, then $K=K^*$ and the
twicing kernel is $K+K^*-KK^*=K$, and nothing changes,
but in general using a twicing kernel can cut a bias
significantly.

If $K$ is a kernel operator with kernel $(x_1,x_2)\mapsto K(x_1,x_2)$,
then the adjoint operator is a kernel operator
with kernel $(x_1,x_2)\mapsto K(x_2,x_1)$, and
the twicing operator $K+K^*-KK^*$ is a kernel operator with kernel
(which depends on $G$)
\begin{equation}
(x_1,x_2)\mapsto K(x_1,x_2)+K(x_2,x_1)
-\int K(x_1,z)K(x_2,z)\,dG(z).
\label{EqTwicingKernel}\end{equation}

\beginskip
\subsection{Haar}
\label{ExampleHaar}
For  ``father'' and  ``mother'' functions $\phi: \RR\to\RR$
and $\psi: \RR\to\RR$ defined by
$\phi(x)=1_{(0,1]}(x)$ and $\psi(x)=1_{(0,1/2]}(x)-1_{(1/2,1]}(x)$,
the {\sl Haar basis} is the set of functions $\{\phi, \psi_{i,j}: 
i=0,1,\ldots, j=0,1,2,\ldots, 2^i-1\}$, for
$$\psi_{i,j}(x)=2^{i/2}\psi(2^ix-j),
\qquad i=0,1,\ldots,\quad j=0,1,2,\ldots, 2^i-1.$$
The Haar basis is a complete orthonormal basis of $L_2(\LEB)$
for $\LEB$ the Lebesgue measure on $[0,1]$.
The linear span of  the first $k:=2^I$ basis elements
$\{\phi, \psi_{i,j}: i=0,1,\ldots,I-1; j=0,1,2,\ldots, 2^i-1\}$
is equal to the linear span of the scaled and shifted father
functions $\{\phi_{I,j}: j=0,1,2\ldots, 2^I-1\}$, given by
$$\phi_{i,j}(x)=2^{i/2}\phi(2^ix-j),
\qquad i=0,1,\ldots,\quad j=0,1,2,\ldots, 2^i-1.$$
These father functions at level $I$ are themselves an orthonormal
system and hence (cf.\ (\ref{EqProjectionKernelOnBasis}))
the orthogonal projection on their span
is a kernel operator with kernel
\begin{align*}
K_k(x_1,x_2)&=\sum_{j=0}^{2^I-1}\phi_{I,j}(x_1) \phi_{I,j}(x_2)\\
&=k \sum_{j=0}^{2^I-1}1_{(j2^{-I}, (j+1)2^{-I}]}(x_1)
1_{(j2^{-I},(j+1)2^{-I}]}(x_2).
\end{align*}
The functions $\phi_{I,j}$ are orthonormal relative to
the Lebesgue measure, but also orthogonal relative to
any other measure $G$. The kernel of the $L_2(G)$-projection 
on the Haar basis, as in (\ref{EqProjectionOnFixedSpace}), is obtained
by simply inserting the inverses $A_{j,j}$  of
the $L_2(G)$-norms of the basis functions in the sum 
($A_{i,,j}=0$ for $i\not=j$). Thus we consider kernels
of the form, with $A_{j,j}$ constants bounded away from zero and
infinity,
\begin{equation}
K_k(x_1,x_2)=k \sum_{j=0}^{2^I-1}A_{j,j}1_{(j2^{-I}, (j+1)2^{-I}]}(x_1)
1_{(j2^{-I},(j+1)2^{-I}]}(x_2).\label{EqHaarKernel}\end{equation}
Clearly the support of the kernel is a set of the form
shown in Figure~\ref{FigurePartition}. Condition (\ref{EqConditionOne})
is trivial if the partitions $\X=\cup_m\X_{n,m}$ are chosen
a coarsening of the partition $(0,1]=\cup_j (j2^{-I},(j+1)2^{-I}]$.
In order to satisfy (\ref{EqConditionTwoHalf})-(\ref{EqConditionThree})
we use a true coarsening of the ``natural'' partition
in fewer sets.

\begin{proposition}
\label{TheoremHaar}
For the Haar kernel (\ref{EqHaarKernel}) with $k=k_n=2^I$ satisfying
$k_n/n\ra \infty$ and $k_n/n^2\ra 0$
conditions (\ref{EqBoundedNorms}), (\ref{EqConditionOne}),
(\ref{EqConditionTwo}), (\ref{EqConditionTwoHalf}), (\ref{EqConditionThree})
and (\ref{EqConditionFour}) are satisfied for any measure
$G$ on $[0,1]$ with a Lebesgue density that is bounded 
and bounded away from zero and regression functions
$\m_2$ and $\m_q$ (for some $q>2$) 
that are bounded and bounded away from zero.
\end{proposition}

\endskip

\subsection{Wavelets}
Consider expansions of functions $f\in L_2(\RR^d)$ on an orthonormal basis 
of compactly supported, bounded wavelets of the form
\begin{equation}
f(x)=\sum_{j\in\ZZ^d}\sum_{v\in\{0,1\}^d}
\!\langle f,\psi_{0,j}^v\rangle \psi_{0,j}^v(x)
+\sum_{i=0}^\infty\sum_{j\in\ZZ^d}\sum_{v\in\{0,1\}^d-\{0\}}
\!\!\langle f,\psi_{i,j}^v\rangle \psi_{i,j}^v(x),
\label{EqWaveletExpansionf}\end{equation}
where the base functions $\psi_{i,j}^v$ are orthogonal
for different indices $(i,j,v)$ and are scaled and translated versions
of the $2^d$ base functions $\psi_{0,0}^v$:
$$\psi_{i,j}^v(x)=2^{id/2}\psi_{0,0}^v(2^i x-j).$$
Such a higher-dimensional wavelet basis can be
obtained as tensor products 
$\psi_{0,0}^v=\phi^{v_1}\times\cdots\times\phi^{v_d}$
of a given father wavelet $\phi^0$ and and mother wavelet $\phi^1$
in one dimension. See for instance Chapter~8 of \cite{Daubechies}.

We shall be interested in functions $f$ with support $\X=[0,1]^d$.
In view of the compact support of the wavelets, for each
resolution level $i$ and vector $v$ only to the order $2^{id}$ base elements
$\psi_{i,j}^v$ are nonzero on $\X$; denote the corresponding set
of indices $j$ by $J_i$.
Truncating the expansion at the level of resolution $i=I$ then gives
an orthogonal projection on a subspace of dimension $k$ of the order 
$2^{Id}$. The corresponding kernel is
\begin{align}
K_k(x_1,x_2)
&=\sum_{j\in J_0}\sum_{v\in\{0,1\}^d}
\!\psi_{0,j}^v(x_1)\psi_{0,j}^v(x_2) \label{EqWaveletKernel}\\
&\qqqquad+\sum_{i=0}^{I}\sum_{j\in J_i}\sum_{v\in\{0,1\}^d-\{0\}}
\!\! \psi_{i,j}^v(x_1) \psi_{i,j}^v(x_1).\nonumber
\end{align}

\begin{proposition}
\label{TheoremWavelets}
For the wavelet kernel (\ref{EqWaveletKernel}) with $k=k_n=2^{Id}$ satisfying
$k_n/n\ra \infty$ and $k_n/n^2\ra 0$
conditions (\ref{EqBoundedNorms}), (\ref{EqConditionOne}),
(\ref{EqConditionTwo}), (\ref{EqConditionTwoHalf}), (\ref{EqConditionThree})
and (\ref{EqConditionFour}) are satisfied for any measure
$G$ on $[0,1]^d$ with a Lebesgue density that is bounded 
and bounded away from zero and regression functions
$\m_2$ and $\m_q$ (for some $q>2$) 
that are bounded and bounded away from zero.
\end{proposition}


\subsection{Fourier basis}
Any function $f\in L_2[-\pi,\pi]$ can be represented 
through the Fourier series $f=\sum_{j\in \ZZ} f_j e_j$,
for the functions $e_j(x)=e^{ij x}/\sqrt{2\pi}$ and
the Fourier coefficients $f_j=\int_{-\pi}^\pi f e_j\,d\LEB$.
The truncated series 
$f_k=\sum_{|j|\le k} f_j e_j$ gives the
orthogonal projection of $f$ onto the linear span
of the function $\{e_j: |j|\le k\}$, and can be written
as $K_kf$ for $K_k$ the kernel operator with kernel
(known as the Dirichlet kernel)
\begin{equation}
K_k(x_1,x_2)=\sum_{|j|\le k} e_j(x_1)e_j(x_2)
={\sin\bigl((k+\thalf)(x_1-x_2)\bigr)\over 2\pi\sin\bigl(\thalf(x_1-x_2)\bigr)}.
\label{EqFourierKernel}\end{equation}

\begin{proposition}
\label{TheoremFourier}
For the Fourier kernel (\ref{EqFourierKernel}) with $k=k_n$ satisfying
$n\ll k_n\ll n^2$ 
conditions (\ref{EqBoundedNorms}), (\ref{EqConditionOne})--(\ref{EqConditionFour}) 
are satisfied for any measure
$G$ on $\RR$ with a bounded Lebesgue density and regression functions
$\m_2$ and $\m_q$ (for some $q>2$) 
that are bounded and bounded away from zero.
\end{proposition}

\subsection{Convolution}
For a uniformly bounded function 
$\phi: \RR\to \RR$ with $\int |\phi|\,d\LEB<\infty$,
and a positive number $\s$, set
\begin{equation}
K_\s(x_1,x_2)={1\over \s}\phi\Bigl({x_1-x_2\over \s}\Bigr)
:=\phi_\s(x_1-x_2).
\label{EqConvolutionKernel}\end{equation}
For $\s\da0$ these kernels tend to the diagonal, with square
norm of the order $\s^{-1}$.

\begin{proposition}
\label{TheoremConvolution}
For the convolution kernel (\ref{EqConvolutionKernel}) with $\s=\s_n$ satisfying
$n^{-2}\ll \s_n\ll n^{-1}$ 
conditions (\ref{EqBoundedNorms}), (\ref{EqConditionOne})--(\ref{EqConditionFour}) 
are satisfied for any measure
$G$ on $[0,1]$ with a Lebesgue density that is bounded 
and bounded away from zero and regression functions
$\m_2$ and $\m_q$ (for some $q>2$) 
that are bounded and bounded away from zero.
\end{proposition}

\subsection{Splines}
The  {\sl Schoenberg space\/} $S_r(T,d)$ of order $r$ for a given
knot sequence $T: t_0=0<t_1<t_2<\cdots<t_l<1=t_{l+1}$ and 
vector of {\sl defects} $d=(d_1,\ldots, d_l)\in\{0,\ldots, r-1\}$ 
are the functions $f: [0,1]\ra\RR$ whose restriction
to each subinterval $(t_i,t_{i+1})$ is a polynomial 
of degree $r-1$ and which are $r-1-d_i$ times
continuously differentiable in a neighbourhood of each $t_i$.
(Here ``0 times continuously differentiable'' means
``continuous'' and ``-1 times continuously differentiable'' 
means no restriction.) The Schoenberg space is a
$k=r+\sum_i d_i$-dimensional vector space. Each ``augmented knot sequence''
\begin{equation}
-t_{r+1}\le\cdots\le t_0=0<t_1<t_2<\cdots<t_l<1=t_{l+1}\le\cdots\le
t_{l+r}\label{EqKnotSequence}\end{equation}
defines a basis $N_1,\ldots, N_k$ of {\sl B-splines}.
These are nonnegative splines with $\sum_j N_j=1$ such that
$N_j$ vanishes outside the interval $(t_j',t_{j+r}')$.
Here the ``basic knots''
$(t_j')$ are defined as the knot sequence $(t_j)$, but with each $t_i\in (0,1)$
repeated $d_i$ times.
See \cite{DevoreLorentz}, pages 137, 140 and 145).
We assume that $|t_{i-1}-t_i|\le |t_{-1}-t_0|$ if $i<0$ and
$|t_{i+1}-t_i|\le |t_{l+1}-t_l|$ if $i>l$.

The {\sl quasi-interpolant operator} is a projection
$K_k: L_1[0,1]\to S_r(T,d)$ with the properties
\begin{align*}
\|f- K_kf\|_p&\le C_r\|f-S_r(T,d)\|_p,\\
\|K_kf\|_p&\le C_r\|f\|_p.
\end{align*}
for every $1\le p\le\infty$ and 
a constant $C_r$ depending on $r$ only 
(see \cite{DevoreLorentz}, pages 144--147).
It follows that the projection $K_k$ inherits the good approximation properties
of spline functions, relative to any $L_p$-norm. In particular,
it gives good approximation to smooth functions.

The quasi-interpolant operator $K_k$ is a projection onto $S_r(T,d)$
(i.e.\ $K_k^2=K_k$ and $K_kf=f$ for $f\in S_r(T,d)$),
but not an orthogonal projection. Because the B-splines form
a basis for $S_r(T,d)$, the operator can be
written in the form $K_kf=\sum_j c_j(f)N_j$ for certain linear functionals
$c_j: L_1[0,1]\to \RR$. It can be shown that, for any $1\le p\le\infty$,
\begin{equation}
|c_j(f)|\le C_r {1\over (t_{j+r}'-t_j')^{1/p}}\|f1_{[t_j',t_{j+r}']}\|_p.
\label{EqCj}\end{equation}
(\cite{DevoreLorentz}, page 145.)
In particular, the functionals $c_j$ belong to the dual space
of $L_1[0,1]$ and can be written as $c_j(f)=\int fc_j\,d\LEB$ for
(with abuse of notation)
certain functions $c_j\in L_\infty[0,1]$. This yields the 
representation of $K_k$ as a kernel operator with kernel
\begin{equation}
K_k(x_1,x_2)=\sumjk N_j(x_1)c_j(x_2).
\label{EqSplineKernel}\end{equation}

\begin{proposition}
\label{TheoremSplines}
Consider a  sequence (indexed by $l$)
of augmented knot sequences (\ref{EqKnotSequence}) with 
$l^{-1}\lesssim t_{i+1}^l-t_i^l\lesssim l^{-1}$ for every $0\le i\le l$
and splines with fixed defects $d_i=d$.
For the corresponding (symmetrized) spline kernel (\ref{EqSplineKernel})
with $l=l_n$ conditions (\ref{EqBoundedNorms}), (\ref{EqConditionOne}),
(\ref{EqConditionTwo}), (\ref{EqConditionTwoHalf}), (\ref{EqConditionThree})
and (\ref{EqConditionFour}) are satisfied if  
$l_n/n\ra \infty$ and $l_n/n^2\ra 0$ for any measure
$G$ on $[0,1]$ with a Lebesgue density that is bounded 
and bounded away from zero and regression functions
$\m_2$ and $\m_q$ (for some $q>2$) 
that are bounded and bounded away from zero.
\end{proposition}

\section{Proof of Theorem~\ref{TheoremMainResult}}
\label{SectionProofs}
For $M_n$ the cardinality of the partition $\X=\cup_m\X_{n,m}$, 
let $N_{n,1},\ldots, N_{n,M_n}$ be the numbers of $X_r$ falling in the partitioning sets, i.e.\
\begin{align*}
I_{n,r}&=m\quad\hbox{ if }X_r\in \X_{n,m},\\
N_{n,m}&=\#(1\le r\le n: I_{n,r}=m).
\end{align*}
The vector $N_n=(N_{n,1},\ldots, N_{n,M_n})$ is multinomially
distributed with parameters $n$ and vector of success
probabilities $p_n=(p_{n,1},\ldots, p_{n,M_n})$ given by
$$p_{n,m}=G(\X_{n,m}).$$ 
Given the vector $I_n=(I_{n,1},\ldots, I_{n,n})$ 
the vectors $(X_1,Y_1),\ldots,(X_n,Y_n)$ are independent
with distributions determined by 
\begin{align}
&\hbox{$X_r$ has distribution $G_{n,I_{n,r}}$ given by 
$dG_{n,I_{n,r}}=1_{\X_{n,I_{n,r}}}\,dG/p_{n,I_{n,r}}$}\label{EqCondisbutOne}\\
&\text{$Y_r$ has the same conditional distribution given $X_r$ as before.}
\label{EqCondisbutTwo}
\end{align}
\par\noindent
We define $U$-statistics $V_n$ by restricting the kernel
$K_n$ to the set $\cup_m\X_{n,m}\times\X_{n,m}$, as follows:
\begin{equation}
V_n={1\over n(n-1)}\dsum_{1\le r\not=s\le n}K_n(X_r,X_s)Y_rY_s
1_{(X_r,X_s)\in \cup_m\X_{n,m}\times\X_{n,m}}.
\label{EqV}\end{equation}
The proof of Theorem~\ref{TheoremMainResult} consists of three elements.
We show that the difference between $U_n$ and $V_n$ is asymptotically
negligible due to the fact that the kernels shrink to the diagonal,
we show that the statistics $V_n$ are conditionally asymptotically
normal given the vector of bin indicators $I_n$, and we
show that the conditional and unconditional means and variances of 
$V_n$ are asymptotically equivalent. These three elements
are expressed in the following four lemmas, which 
should be understood all implicitly to 
assume the conditions of Theorem~\ref{TheoremMainResult}.

\begin{lemma}
\label{LemmaOne}
$\var (U_n-V_n)/\var U_n\ra 0$.
\end{lemma}

\begin{lemma}
\label{LemmaTwo}
$\sup_x\bigl|\Pr(\bigl(V_n-\E(V_n\given I_n)\bigr)/
\sd(V_n\given I_n)\le x\given I_n\bigr)-\Phi(x)\bigr|\prob0$.
\end{lemma}

\begin{lemma}
\label{LemmaThree}
$\bigl(\E V_n-\E(V_n\given I_n)\bigr)/\sd V_n\prob 0$.
\end{lemma}

\begin{lemma}
\label{LemmaFour}
$\var (V_n\given I_n)/\var V_n\prob 1$.
\end{lemma}

\subsection{Proof of Theorem~\ref{TheoremMainResult}}
By Lemmas~\ref{LemmaOne} and~\ref{LemmaThree} the sequence $\bigl((U_n-\E U_n)-(V_n-\E (V_n\given I_n)\bigr)/\sd V_n$
tends to zero in probability. Because conditional and unconditional convergence in probability
to a constant is the same, we see that it suffices to show that
$(V_n-\E (V_n\given I_n)\bigr)/\sd V_n$ converges conditionally given $I_n$ to the normal
distribution, in probability. This follows from Lemmas~\ref{LemmaFour} and~\ref{LemmaTwo}.

The variance of $U_n$ is computed in (\ref{EqVarU}) in Section~\ref{SectionMoments}.
By the Cauchy-Schwarz inequality 
(cf.\ (\ref{EqKernelOperator})),
\begin{align*}
\langle K_n\m,\m\rangle_G^2
&\le \|K_n\m\|_G^2\|\m\|_G^2\le\|K_n\|^2\|\m\|_G^4,\\
\|(K_n\m)\sqrt{\m_2}\|_G^2
&\le \|\m_2\|_\infty\|K_n\m\|_G^2\le\|\m_2\|_\infty\|K_n\|^2\|\m\|_G^2.
\end{align*}
Because $\m_2$ is bounded by assumption and
the norms $\|K_n\|$ are bounded in $n$ by assumption (\ref{EqBoundedNorms}), 
the right sides are bounded in $n$.
In view of (\ref{EqKnToInfinity}) it follows that the first
two terms in the final expression for the variance are of lower order
than the third, whence
\begin{equation}
\var U_n\sim {2k_n\over n^2}.
\label{EqVarUKns}\end{equation}

\subsection{Moments of $U$-statistics}
\label{SectionMoments}
To compute or estimate moments of $U_n$ we employ the Hoeffding decomposition
(e.g.\ \cite{vdVAS}, Sections~11.4 and~12.1)
$U_n=\E U_n+U_n^{(1)}+U_n^{(2)}$ of $U_n$ given by 
\begin{align}
U_n^{(1)}&={2\over n} \sum_{r=1}^n \bigl(K_n\m(X_r)Y_r-\E U_n\bigr),
\label{EqHoeffding}\\
\nonumber
U_n^{(2)}&=\textstyle{1\over n(n-1)}\dsum_{1\le r\not=s\le n}
\Bigl[K_n(X_r,X_s)Y_rY_s-K_n\m(X_r)Y_r-K_n\m(X_s)Y_s+\E U_n\Bigr].
\nonumber
\end{align}
The variables $U_n^{(1)}$ and $U_n^{(2)}$ are uncorrelated,
and so are all the variables in the single and double sums
defining $U_n^{(1)}$ and $U_n^{(2)}$. It follows that
\begin{align}
\var U_n
&={4\over n} \var\bigl(K_n\m(X_1)Y_1\bigr) \label{EqVarU} \\
&\quad
+{2\over n(n-1)}\var\bigl(K_n(X_1,X_2)Y_1Y_2-K_n\m(X_1)Y_1-K_n\m(X_2)Y_2
\bigr)\nonumber\\
&=\Bigl[{4\over n}-{4\over n(n-1)}\Bigr] \var\bigl(K_n\m(X_1)Y_1\bigr)
+{2\over n(n-1)}\!\var\bigl(K_n(X_1,X_2)Y_1Y_2\bigr)
\nonumber\\
&={4(n-2)\over n(n-1)} \|(K_n\m)\sqrt{\m_2}\|_G^2
-{4(n-2)+2\over n(n-1)} \langle K_n\m,\m\rangle_G^2
+{2k_n\over n(n-1)}
\nonumber
\end{align}
See equation (\ref{EqDefinitionkn}) for the definition of $k_n$.

There is no similarly simple expression for higher moments of
a $U$-statistic, but the following useful bound is 
(essentially) established in \cite{GineLatalaZinn}.

\begin{lemma}
[Gin\'e, Latala, Zinn]
\label{LemmaGLZ}
For any $q\ge 2$ there exists a constant $C_q$ such that
for any i.i.d.\ random variables $X_1,\ldots,X_n$ and degenerate
symmetric kernel $K$,
\begin{align*}
&\E \Bigl|{1\over n(n-1)}\dsum_{1\le r\not=s\le n}K(X_r,X_s)\Bigr|^q\\
&\qquad \le C_q n^{-q}\,\Bigl(\E K^2(X_1,X_2)\Bigr)^{q/2}\vee
n^{-3q/2+1}\, \E |K(X_1,X_2)|^q\\
&\qquad\le C_q n^{-q}\,\E |K(X_1,X_2)|^q.
\end{align*}
\end{lemma}

\begin{proof}
The second inequality is immediate from the fact that
the $L_2$-norm is bounded above by the $L_q$-norm,
and $3q/2-1\ge q$,  for $q\ge 2$.
For the first inequality we use (3.3) in \cite{GineLatalaZinn}
(and decoupling as explained in Section 2.5 of that paper) to see that
the left side of the lemma is bounded above by a multiple of 
\begin{align*}
&n^{-q}\,\Bigl(\E K^2(X_1,X_2)\Bigr)^{q/2}\vee
n^{-3q/2+1}\, \E\bigl(\E (K^2(X_1,X_2)\given X_2)\bigr)^{q/2}\\
&\qqqquad\qqqquad\qqqquad\qqqquad
\vee n^{2-2q}\,\E |K(X_1,X_2)|^q.
\end{align*}
Because $L_q$-norms are increasing in $q$,
 the second term on the right is bounded
above by $n^{-3q/2+1}\E |K(X_1,X_2)|^q$, which is also a bound
on the third term, as $n^{2-2q}\le n^{-3q/2+1}$ for $q\ge 2$.
\end{proof}

We can apply the preceding inequality to the degenerate part
of the Hoeffding decomposition (\ref{EqHoeffding}) of $U_n$ and combine it
with the Marcinkiewicz-Zygmund inequality
to obtain a bound on the moments of $U_n$.

\begin{corollary}
\label{CorollaryMoments}
For any $q\ge 2$ there exists a constant $C_q$ such that
for the $U$-statistic given by (\ref{EqU}) and (\ref{EqHoeffding}),
\begin{align*}\E |U_n^{(1)}|^q
&\le C_q\, n^{-q/2}\int |K_n\m|^q\,\m_q\,dG,\\
\E |U_n^{(2)}|^q
&\le C_q\,
n^{-q}\,\Bigl(\int\!\!\int K_n^2\ \m_2\times\m_2\,dG\times G\Bigr)^{q/2}
\\
&\qqqquad\qqqquad
\vee C_q\,n^{-3q/2+1}\, \int\!\!\int |K_n|^q\ \m_q\times\m_q\,dG\times G.
\end{align*}
\end{corollary}

\begin{proof}
The first inequality follows from the Marcinkiewicz-Zygmund
inequality and the fact that $\E|Z-\E Z|^q\le 2^q \E |Z|^q$, for any random variable $Z$. 
To obtain the second  we apply 
Lemma~\ref{LemmaGLZ} to $U_n^{(2)}$, which is a  degenerate $U$-statistic
with kernel $K_n(X_1,X_2)Y_1Y_2-\Pi_n(X_1,X_2,Y_1,Y_2)$, for
$\Pi_n$ the sum of the conditional expectations of
$K_n(X_1,X_2)Y_1Y_2$ relative to $(X_1,Y_1)$ and $(X_2,Y_2)$
minus $\E U_n$. Because (conditional) expectation 
is a contraction for the $L_q$-norm ($\E \bigl|\E(Z\given \A)\bigr|^q\le
\E |Z|^q$ for any random variable $Z$ and conditioning $\s$-field
$\A$), we can bound the $L_2$- and $L_q$-norms of the 
degenerate kernel, appearing in the bound
obtained from Lemma~\ref{LemmaGLZ}, by a constant (depending on $q$) 
times the $L_2$- of $L_q$-norm of the kernel $K_n(X_1,X_2)Y_1Y_2$.
\end{proof}


\subsection{Proof of Lemma~\ref{LemmaOne}}
The statistic $U_n-V_n$ is a $U$-statistic of the same type as $U_n$,
except that the kernel $K_n$ is replaced by $K_n(1-1_{\X_n})$ for
$\X_n=\cup_m(\X_{n,m}\times\X_{n,m})$. The variance of $U_n-V_n$ is
given by formula (\ref{EqVarU}), but with $K_n$ replaced by 
the kernel operator with kernel $K_{n,n}=K_n(1-1_{\X_n})$. 
The corresponding kernel operator
is $K_{n,n}f= K_nf-\sum_mK_n(f1_{\X_{n,m}})1_{\X_{n,m}}$,
and hence 
\begin{align*}
\thalf\|K_{n,n}f\|_G^2
&\le \|K_nf\|_G^2+\Bigl\|\sum_m K_n(f1_{\X_{n,m}})1_{\X_{n,m}}\Bigr\|_G^2\\
&\le\|K_nf\|_G^2+\sum_m\|K_n(f1_{\X_{n,m}})\|_G^2\\
&\le\|K_n\|^2\|f\|_G^2+\sum_m\|K_n\|^2\|f1_{\X_{n,m}}\|_G^2
\le2\|K_n\|^2\|f\|_G^2.
\end{align*}
It follows that the operator norms $\|K_{n,n}\|_2$ of the operators
$K_{n,n}$  are uniformly bounded in $n$ (cf.\ equation (\ref{EqBoundedNorms})
for the operators $K_n$).
Applying decomposition (\ref{EqVarU}) to the kernel $K_{n,n}$
we see that $\var(U_n-V_n)=O(n^{-1})+2k_{n,n}/n^2$,
where $k_{n,n}$ is the $L_2(G\times G)$-norm
$k_{n,n}$ of the kernel $K_{n,n}$ weighted by $\m_2\times\m_2$,
as in (\ref{EqDefinitionkn}) but with $K_n$ replaced by $K_{n,n}$.
By assumption (\ref{EqConditionOne}) the norm
$k_{n,n}$ is negligible relative to the same norm  
(denoted $k_n$) of the original kernel.
Because the variance of $U_n$ is asymptotically 
equivalent to $2k_n/n^2$ and $k_n/n\ra\infty$, this proves
the claim.

\subsection{Proof of Lemma~\ref{LemmaTwo}}
The variable $V_n$ can be written as the sum $V_n=\sum_m V_{n,m}$, for
\begin{equation}
V_{n,m}={1\over n(n-1)}\dsum_{1\le r\not=s\le n}K_n(X_r,X_s)Y_rY_s
1_{(X_r,X_s)\in \X_{n,m}\times\X_{n,m}}.
\label{EqVm}
\end{equation}
Given the vector of bin-indicators $I_n$ the observations $(X_r,Y_r)$
are independently generated from the conditional distributions
in which $X_r$ is conditioned to fall in bin $\X_{n,I_{n,r}}$, 
as given in (\ref{EqCondisbutOne})-(\ref{EqCondisbutTwo}).
Because each variable $V_{n,m}$ depends only on the observations $(X_r,Y_r)$
for which $X_r$ falls in bin $\X_{n,m}$, the variables
$V_{n,1},\ldots, V_{n,M_n}$ are conditionally independent. The conditional
asymptotic normality of $V_n$ given $I_n$ can therefore be
established by a central limit theorem for independent
variables.

The variable $V_{n,m}$ is equal to 
$N_{n,m}(N_{n,m}-1)/\bigl(n(n-1)\bigr)$ times
a $U$-statistic of the type (\ref{EqU}), based on $N_{n,m}$ observations
$(X_r,Y_r)$ from the conditional distribution where $X_r$ is conditioned
to fall in $\X_{n,m}$. The corresponding kernel operator
is given by
\begin{align}
K_{n,m}f(x)
&=\int K_n(x,v)f(v)1_{\X_{n,m}\times\X_{n,m}}(x,v)\,{dG(v)\over p_{n,m}}
\nonumber\\
&={K(f1_{\X_{n,m}})(x)1_{\X_{n,m}}(x)\over p_{n,m}}.
\label{EqDefKnm}\end{align}
We can decompose each $V_{n,m}$ into
its Hoeffding decomposition $V_{n,m}=\E(V_{n,m}\given I_n)
+V_{n,m}^{(1)}+V_{n,m}^{(2)}$ relative to the conditional 
distribution given $I_n$. We shall show that
\begin{equation}
\E \left({|\sum_m V_{n,m}^{(1)}|\over \sd (V_n\given I_n)}
\given I_n\right)\prob 0.
\label{EqVnOneToZero}\end{equation}
To prove Lemma~\ref{LemmaTwo} it then suffices to show
that  the sequence $\sum_m V_{n,m}^{(2)}/ \sd (V_n\given I_n)$
converges conditionally given $I_n$ weakly to the standard normal distribution,
in probability. By Lyapounov's theorem, this follows from,
for some $q>2$,
\begin{equation}
{\sum_m \E\bigl(|V_{n,m}^{(2)}|^q\given I_n\bigr)
\over \sd(V_n\given I_n)^{q}}
\prob 0.\label{EqLyapounov}\end{equation}
By Lemma~\ref{LemmaFour} the conditional standard deviation
$\sd(V_n\given I_n)$ is asymptotically equivalent in probability
to the unconditional standard deviation, and
by Lemma~\ref{LemmaOne} this is equivalent to $\sd U_n$, which
is equivalent to $\sqrt{2k_n/n^2}$. Thus in both
(\ref{EqVnOneToZero}) and (\ref{EqLyapounov}) the conditional
standard deviation in the denominator
may be replaced by $\sqrt{2k_n/n^2}$.
 
In view of the first assertion of Corollary~\ref{CorollaryMoments},
$$\var (V_{n,m}^{(1)}\given I_n)\le 
C_2 \Bigl({N_{n,m}(N_{n,m}-1)\over n(n-1)}\Bigr)^2 \,  N_{n,m}^{-1}
\int \left|{K_n(\m 1_{\X_{n,m}})\over p_{n,m}}\right|^2\m_2\,
1_{\X_{n,m}}{dG\over p_{n,m}}.$$
By Lemma~\ref{LemmaMomentBinomial} (below,
note that $(np_{n,m})^2\lesssim (np_{n,m})^3$ in view
of (\ref{EqConditionThree})) the expectation of the
right side is bounded above by a constant times
$${(np_{n,m})^3\over n^2(n-1)^2p_{n,m}^3}
\,\|\m_2\|_\infty\,\|K_n(\m 1_{\X_{n,m}})\|_G^2
\le {1\over n}\,\|\m_2\|_\infty\,\|K_n\|^2\,\|\m 1_{\X_{n,m}}\|_G^2.$$
In view of (\ref{EqBoundedNorms}) 
the sum over $m$ of this expression is bounded
above by a multiple of $1/n$, which is $o(k_n/n^2)$ by
assumption (\ref{EqKnToInfinity}). Because $\E(V_{n,m}^{(1)}\given I_n)=0$,
this concludes the proof of 
(\ref{EqVnOneToZero}).

In view of the second assertion of Corollary~\ref{CorollaryMoments},
\begin{align*}
\E |V_{n,m}^{(2)}|^q
&\le  C_q\,\Bigl({N_{n,m}(N_{n,m}-1)\over n(n-1)}\Bigr)^q\times\\
&\biggl[
N_{n,m}^{-q}\,\Bigl(\int\!\!\int K_n^2\ \m_2\times\m_2\,
1_{\X_{n,m}\times\X_{n,m}}\,{dG\times G\over p_{n,m}^2}\Bigr)^{q/2}\\
&\qqqquad\qquad
\vee N_{n,m}^{-3q/2+1}\,\int\!\!\int |K_n|^q\ \m_q\times\m_q\,
1_{\X_{n,m}\times\X_{n,m}}\,{dG\times G\over p_{n,m}^2}\biggr].
\end{align*}
By Lemma~\ref{LemmaMomentBinomial} the expectation of the
right side is bounded above by a constant times
\begin{align*}
&{(np_{n,m})^q\over n^q(n-1)^qp_{n,m}^q}\,
\Bigl(\int\!\!\int K_n^2\ \m_2\times\m_2\,
1_{\X_{n,m}\times\X_{n,m}}\,dG\times G\Bigr)^{q/2}\\
&\qqqquad+ {(np_{n,m})^{q/2+1}\over n^q(n-1)^qp_{n,m}^2}\,
\int\!\!\int |K_n|^q\ \m_q\times\m_q\,
1_{\X_{n,m}\times\X_{n,m}}\,dG\times G.
\end{align*}
With $\a_{n,m}(q)=\int\!\!\int |K_n|^q\ \m_q\times\m_q\,
1_{\X_{n,m}\times\X_{n,m}}\,dG\times G$ it follows that 
\begin{align*}
\sum_m{\E |V_{n,m}^{(2)}|^q\over (k_n/n^2)^{q/2}}
&\lesssim\sum_m\Bigl({\a_{n,m}(2)\over k_n}\Bigr)^{q/2}
+ \sum_m\Bigl({p_{n,m}\over n}\Bigr)^{q/2-1}\,{\a_{n,m}(q)\over k_n^{q/2}}\\
&\minqquad\minqquad\minqquad\le\max_m\Bigl({\a_{n,m}(2)\over k_n}\Bigr)^{q/2-1}
\sum_m{\a_{n,m}(2)\over k_n}
+ \max_m\Bigl({p_{n,m}\over n}\Bigr)^{q/2-1}
\sum_m\,{\a_{n,m}(q)\over k_n^{q/2}}.
\end{align*}
The right side tends to zero
by assumptions (\ref{EqConditionOne}),
(\ref{EqConditionTwo}) and (\ref{EqConditionFour}).
This concludes the proof of (\ref{EqLyapounov}).

\subsection{Proof of Lemma~\ref{LemmaThree}}
Only pairs $(X_r,X_s)$ that fall in one of the sets
$\X_{n,m}\times \X_{n,m}$ contribute to the double sum
(\ref{EqV}) that defines $V_n$. Given $I_n$ there are
$N_{n,m}(N_{n,m}-1)$ pairs that fall in $\X_{n,m}$
and the distribution of the corresponding 
vectors $(X_r,Y_r), (X_s,Y_s)$
is determined as in (\ref{EqCondisbutOne})-(\ref{EqCondisbutTwo}).
From this it follows that 
$$\E(V_n\given I_n)
={1\over n(n-1)}\sum_m
N_{n,m}(N_{n,m}-1)\int\!\int\!\!\!K_n\,\m\times\m\,1_{\X_{n,m}\times\X_{n,m}}
{dG\times G\over p_{n,m}^2}.$$
Defining the numbers 
$\a_{n,m}=\int\int K_n\ \m\times\m\,1_{\X_{n,m}\times\X_{n,m}}\,dG\times G$,
we infer that 
$$\E(V_n\given I_n)-\E V_n=
\sum_m\Bigl({N_{n,m}(N_{n,m}-1)\over n(n-1)p_{n,m}^2}-1\Bigr)\a_{n,m}.$$
By the Cauchy-Schwarz inequality, the numbers $\a_{n,m}$ satisfy
$$|\a_{n,m}|\le \|K_n(\m1_{\X_{n,m}})\|_G\|\m1_{\X_{n,m}}\|_G
\le \|K_n\|\|\m1_{\X_{n,m}}\|_G^2\lesssim \|K_n\|\|\m\|_\infty^2 p_{n,m}.$$
In particular $\sum_m|\a_{n,m}|\lesssim 1$.
In view of (\ref{EqBoundedNorms}) 
the numbers  $s_n^2$ given in (\ref{EqDefinitionSn}) (below)
are of the order $M_n/n^2+1/n$.
Lemma~\ref{LemmaLLNMultinomialForm} (below)
therefore implies that the right side of the second last
display is of the order $O_P(\sqrt{M_n}/n+1/\sqrt n)=O(1/\sqrt n)$,
because (\ref{EqConditionThree}) implies that $M_n\lesssim n$.
By assumption (\ref{EqKnToInfinity})
this is smaller than $\sqrt {k_n/n^2}$, which is
of the same order as $\sd V_n$.

\subsection{Proof of Lemma~\ref{LemmaFour}}
By (\ref{EqVarU}) applied to the variables $V_{n,m}$ defined in (\ref{EqVm}),
\begin{align*}
\var(V_n\given I_n)
&=\sum_m \var(V_{n,m}\given I_n)\\
&=\sum_m \Bigl({N_{n,m}(N_{n,m}-1)\over n(n-1)}\Bigr)^2
\biggl[{4(N_{n,m}-2)\over N_{n,m}(N_{n,m}-1)}
\|(K_{n,m}\m)\sqrt{\m_2}\|_{G_{n,m}}^2\\
&\qqqquad
-{4(N_{n,m}-2)+2\over N_{n,m}(N_{n,m}-1)}
\langle K_{n,m}\m,\m\rangle_{G_{n,m}}^2
+{2k_{n,m}\over N_{n,m}(N_{n,m}-1)}\biggr],
\end{align*}
where the operator $K_{n,m}$ is given in (\ref{EqDefKnm}),
the distribution $G_{n,m}$ is defined in (\ref{EqCondisbutOne}), and
$$k_{n,m}=\int\!\!\int K_n^2\,\m_2\times\m_2\ 
1_{\X_{n,m}\times\X_{n,m}}\,{dG\times G
\over p_{n,m}^2}=:{\a_{n,m}(2)\over p_{n,m}^2}.$$
We can split this into three terms. 
By Lemma~\ref{LemmaMomentBinomial} the expected value of the
first term is bounded by a multiple of 
$$\sum_m{(np_{n,m})^3\over n^2(n-1)^2p_{n,m}^3}\|\m_2\|_\infty
\|K_n(\m1_{\X_{n,m}})\|_G^2
\le {1\over n}\|\m_2\|_\infty\|K_n\|^2\|\m\|_G^2.$$
Similarly the expected value of the
absolute value of the second term is bounded by a multiple of 
\begin{align*}
\sum_m{(np_{n,m})^3\over n^2(n-1)^2p_{n,m}^4}
\langle K_n(\m1_{\X_{n,m}}),\m1_{\X_{n,m}}\rangle_G^2
&\le \sum_m{1\over np_{n,m}}\|K_n\|^2\|\m1_{\X_{n,m}}\|_G^4\\
&\le {1\over n}\|K_n\|^2\|\m\|_\infty^2\|\m\|_G^2.
\end{align*}
These two terms divided by $k_n/n^2$ tend to zero,
by (\ref{EqKnToInfinity}).

By Lemma~\ref{LemmaOne} and (\ref{EqVarUKns}) 
we have that $\var V_n\sim 2k_n/n(n-1)$, which in term
is asymptotically equivalent to $2\sum_m\a_{n,m}(2)/n(n-1)$,
by (\ref{EqConditionOne}).
It follows that
\begin{align*}\var(V_n\given I_n)-\var V_n
&=2\sum_m{N_{n,m}(N_{n,m}-1)\over n^2(n-1)^2}k_{n,m}-2{k_n\over n(n-1)}
+o\Bigl({k_n\over n^2}\Bigr)\\
&=2\sum_m\Bigl({N_{n,m}(N_{n,m}-1)\over n(n-1)p_{n,m}^2}-1\Bigr){\a_{n,m}(2)
\over n(n-1)}
+o\Bigl({k_n\over n^2}\Bigr).
\end{align*}
Here the coefficients $\a_{n,m}(2)/k_n$ 
satisfy the conditions imposed on $\a_{n,m}$
in Corollary~\ref{CorLLNMultinomialForm}, in view
of  (\ref{EqConditionOne}) and (\ref{EqConditionTwo}).
Therefore this corollary shows that
the expression on the right is $o_P(k_n/n^2)$.

\subsection{Auxiliary lemmas on multinomial variables}

\begin{lemma}
\label{LemmaMomentBinomial}
Let $N$ be binomially distributed with parameters $(n,p)$. For any
$r\ge 2$ there exists a constant $C_r$ such that 
$\E N^r1_{N\ge 2}\le C_r \bigl((np)^r\vee (np)^2\bigr)$.
\end{lemma}

\begin{proof}
\def\rr{{\underline r}}
For $r=\rr+\d$ with $\rr$ an integer and $0\le\d<1$
there exists a constant $C_r$ with
$N^r1_{N\ge 2}\le C_r N^\d N(N-1)\cdots (N-\rr+1)
+C_r N^\d N(N-1)$ for every $N$.
Hence 
\begin{align*}
\E N^r1_{N\ge 2}
&\le C_r \sum_{k= 2}^nk^\d  \Bigl( k(k-1)\cdots (k-\rr+1)
+ k(k-1)\Bigr){n\choose k}p^k(1-p)^{n-k}\\
&=C_r\Bigl((np)^\rr\,\E N_1^\d+(np)^2\, \E N_2^\d\Bigr),
\end{align*}
for $N_1$ and $N_2$ binomially distributed with parameters
$n-\rr$ and $p$ and $n-2$ and $p$, respectively.
By Jensen's inequality $\E N_j^\d\le (\E N_j)^\d$, which
is bounded above by $(n p)^\d$, yielding the
upper bound $C_r \bigl((np)^r+ (np)^{2+\d}\bigr)$. 
If $np\le 1$, then this is bounded above by $2C_r(np)^2$
and otherwise by $2C_r(np)^r$.
\end{proof}

The next result is a law of large numbers 
for a quadratic form in multinomial vectors of increasing
dimension. The proof is based on a comparison
of multinomial variables to Poisson variables
along the lines of the proof of a 
central limit theorem in \cite{Morris}.

\begin{lemma}
\label{LemmaLLNMultinomialForm}
For each $n$ let $N_n$ be multinomially distributed with parameters
$(n,p_{n,1},\ldots, p_{n,M_n})$ 
with $\max_m p_{n,m}\ra 0$ as $n\ra\infty$
and $\liminf_{n\ra\infty }n\min_m p_{n,m}>0$. 
For given numbers $\a_{n,m}$ let
\begin{equation}
s_n^2={2\over n^2}\sum_m {\a_{n,m}^2\over p_{n,m}^2}
+{4\over n}\sum_mp_{n,m}\Bigl({\a_{n,m}\over p_{n,m}}-\sum_m\a_{n,m}\Bigr)^2.
\label{EqDefinitionSn}
\end{equation}
Then 
$$\sum_m \a_{n,m}\Bigl({N_{n,m}(N_{n,m}-1)\over n(n-1)p_{n,m}^2}-1\Bigr)
=O_P\Bigl(s_n+{\sum_m|\a_{n,m}|\over\sqrt n}\Bigr).$$
\end{lemma}

\begin{proof}
Because $\sum_m \a_{n,m}\bigl((n-1)/n-1\bigr)=\sum_m\a_{n,m}(-1/n)$, 
it suffices to prove the statement of the lemma with $n(n-1)$ replaced
by $n^2$. Using the fact that
$\sum_m N_{n,m}=n$ we can rewrite the resulting quadratic form 
as, with $\l_{n,m}=np_{n,m}$,
\begin{align*}
\sum_m\a_{n,m} 
\Bigl(&{N_{n,m}(N_{n,m}-1)\over n^2p_{n,m}^2}-1\Bigr)
=\sqrt 2\sum_m{\a_{n,m}\over\l_{n,m}}C_2(N_{n,m},\l_{n,m})\\
&\qqqquad+2\sum_m\sqrt{\l_{n,m}}\Bigl({\a_{n,m} \over\l_{n,m}}-{\sum_m\a_{n,m}\over n}
\Bigr)C_1(N_{n,m},\l_{n,m}),
\end{align*}
for $C_1$ and $C_2$ the Poisson-Charlier polynomials of degrees
1 and 2, given by
$$C_1(x,\l)={x-\l\over \sqrt \l},\qquad 
C_2(x,\l)={x(x-1)-2\l x+\l^2\over \sqrt 2\l}.$$
Together with $x\mapsto C_0(x)=1$ the functions $x\mapsto C_1(x,\l)$ 
and $x\mapsto C_2(x,\l)$ are the polynomials $1, x, x^2$ orthonormalized
for the Poisson distribution with mean $\l$ by
the Gramm-Schmidt procedure. For $X=(X_1,\ldots,X_{M_n})$ let 
\begin{align*}
T_n(X)&=\sum_m{\a_{n,m}\over\l_{n,m}}C_2(X_m,\l_{n,m})\\
&\qqqquad
+\sum_m\sqrt{2\l_{n,m}}\Bigl({\a_{n,m} \over\l_{n,m}}-{\sum_m\a_{n,m}\over n}
\Bigr)C_1(X_m,\l_{n,m}).
\end{align*}
Thus up to a factor $\sqrt 2$ the statistic
$T_n(N_n)$ is the quadratic form of interest.

If the variables $N_{n,1},\ldots,N_{n,M_n}$ were independent
Poisson variables with mean values $\l_{n,m}$,
then the mean of $T_n(N_n)$ would be zero and the variance would be given 
by $s_n^2/2$, and hence in that case $T_n(N_n)=O_P(s_n)$.
We shall now show that the difference between multinomial and
Poisson variables is of the order $\sum_m|\a_{n,m}|/\sqrt n$.

To make the link between multinomial and Poisson variables,
let $\tilde n$ be a Poisson variable with mean $n$ 
and given $\tilde n=k$ let
$\tilde N_n=(\tilde N_{n,1},\ldots,\tilde N_{n,M_n})$ 
be multinomially distributed with
parameters $k$ and $p_n=(p_{n,1},\ldots, p_{n,M_n})$.
The original multinomial vector $N_n$ is then equal in distribution
to $\tilde N_n$ given $\tilde n=n$.
Furthermore, the vector $\tilde N_n$ is
unconditionally Poisson distributed as in the preceding paragraph,
whence, for any $M_n\ra\infty$,
$$\Pr\bigl(|T_n(\tilde N_n)|>M_ns_n\bigr)\ra 0.$$
The left side is bigger than
\begin{align*}
&\sum_{k: |k-n|\le \sqrt n}
\Pr\bigl(|T_n(\tilde N_n)|>M_ns_n\given\tilde n=k\bigr)\Pr(\tilde n=k)\\
&\qqqquad\ge \min_{k: |k-n|\le \sqrt n}
\Pr\bigl(|T_n(N_n(k))|>M_ns_n\bigr)
\Pr\bigl(|\tilde n-n|\le\sqrt n\bigr),
\end{align*}
where the vector $N_n(k)$ is multinomial 
with parameters $k$ and $p_n$.
Because the sequence $(\tilde n-n)/\sqrt n$ tends to a standard normal
distribution as $n\ra\infty$, the probability 
$\Pr\bigl(|\tilde n-n|\le\sqrt n\bigr)$ tends
to the positive constant $\Phi(1)-\Phi(-1)$.  We conclude that the
sequence of minima on the right tends to zero. The probability of
interest is the term with $k=n$ in the minimum. Therefore the
proof is complete once we show that
the minimum and maximum of the terms are comparable.

To compare the terms with different $k$ we 
couple the multinomial vectors $N_n(k)$ on a single probability space.
For given $k<k'$ we construct these vectors
such that $N_n(k')=N_n(k)+N_n'(k'-k)$ for
$N_n'(k'-k)$ a multinomial vector with parameters $k'-k$ and
$p_n$ independent of $N_n(k)$. For any numbers $N$ and $N'$ we have that
$C_2(N+N',\l)-C_2(N,\l)=\bigl((N')^2+2NN'-N'(1+2\l)\bigr)/(\sqrt 2\l)$. 
Therefore,
\begin{align*}
&\qquad\E \Bigl|\sum_m {\a_{n,m}\over\l_{n,m}}
C_2\bigl(N_{n,m}(k'),\l_{m,n}\bigr)
-\sum_m {\a_{n,m}\over\l_{n,m}}C_2\bigl(N_{n,m}(k),\l_{m,n}\bigr)\Bigr|\\
&\le\sum_m {|\a_{n,m}|\over\l_{n,m}}
{\E\bigl|N_{n,m}'(k'\!\!-k)^2+2N_{n,m}'(k'\!\!-k)N_{n,m}(k)
-N_{n,m}'(k'\!\!-k)(1\!+\!2\l_{n,m})\bigr|\over \sqrt 2\l_{n,m}}.
\end{align*}
For $|k'-n|\le\sqrt n$ and $|k-n|\le\sqrt n$ the binomial variable
$N_{n,m}(k'-k)$ has first and second moment bounded by a multiple
of $\sqrt np_{n,m}$ and $np_{n,m}^2$. From this 
the right side of the display can be seen to 
be of the order $\sum_m|\a_{n,m}|O(n^{-1/2}):=\r_n$.
Similarly, we have $C_1(N+N',\l)-C_1(N,\l)=N'/\sqrt \l$ and
$$\E\Bigl|\sum_m\sqrt{\l_{n,m}}\Bigl({\a_{n,m}\over\l_{n,m}}-
{\sum_m\a_{n,m}\over n}\Bigr)\bigl(C_1(N_{n,m}(k'),\l_{n,m})
-C_1(N_{n,m}(k),\l_{n,m})\bigr)\Bigr|$$
can be seen to be of the order $\sum_m |\a_{n,m}/\l_{n,m}-\sum_m\a_{n,m}/n|
\sqrt np_{n,m}$, which is also of the order 
$\r_n$.

We infer from this that $\E \bigl|T_n(N_n(k))-T_n(N_n(n))\bigr|=O(\r_n)$,
uniformly in $|k-n|\le\sqrt n$, and therefore
\begin{align*}
&\Pr\bigl(|T_n(N_n(n))|>M_n(s_n+\r_n)\bigr)\\
&\qquad \le 
\Pr\bigl(|T_n(N_n(k))|> M_ns_n\bigr)+
\Pr\bigl(|T_n(N_n(n))-T_n(N_n(k))|> M_n\r_n\bigr)\\
&\qquad\le \Pr\bigl(|T_n(N_n(k))|> M_ns_n\bigr)+o(1),
\end{align*}
uniformly in $|k-n|\le\sqrt n$, for every $M_n\ra\infty$,
by Markov's inequality. 
In the preceding paragraph it was seen that the minimum
of the right side over $k$ with $|k-n|\le\sqrt n$ tends to zero
for any $M_n\ra\infty$. Hence so does the left side.
\end{proof}

Under the additional condition that
$${1\over s_n^2}\max_m \biggl[{\a_{n,m}^2\over n^2p_{n,m}^2}
+{p_{n,m}\over n}\Bigl({\a_{n,m}\over p_{n,m}}-\sum_m\a_{n,m}\Bigr)^2\biggr]
\ra 0,$$
it follows from Corollary~4.1 in \cite{Morris} that
the sequence $s_n^{-1}$ times the quadratic form in the
preceding lemma tends in distribution to the standard normal distribution.
Thus in this case the order claimed by the lemma is 
sharp as soon as $n^{-1/2}\sum_m|\a_{n,m}|$ is not bigger than $s_n$.

\begin{corollary}
\label{CorLLNMultinomialForm}
For each $n$ let $N_n$ be multinomially distributed with parameters
$(n,p_{n,1},\ldots, p_{n,M_n})$ 
with $\liminf_{n\ra\infty }n\min_m p_{n,m}>0$. 
If $\a_{n,m}$ are numbers with
$\sum_m|\a_{n,m}|=O(1)$ and $\max_m|\a_{n,m}|\ra 0$ as
$n\ra\infty$, then 
$$\sum_m \a_{n,m}\Bigl({N_{n,m}(N_{n,m}-1)\over n(n-1)p_{n,m}^2}-1\Bigr)
\prob 0.$$
\end{corollary}

\begin{proof}
Since $np_{n,m}\gtrsim 1$ by assumption the numbers $s_n$
defined in (\ref{EqDefinitionSn}) satisfy
\begin{align*}
s_n^2&\le 
2\sum_m {\a_{n,m}^2\over n^2p_{n,m}^2}
+4\sum_m{\a_{n,m}^2\over n p_{n,m}}
\lesssim \sum_m\a_{n,m}^2.
\end{align*}
The corollary is a consequence of Lemma~\ref{LemmaLLNMultinomialForm}.
\end{proof}

\section{Proofs for Section~\ref{SectionExamples}}

\begin{proof}[Proof of Corollary~\ref{CorollaryEstimatorIntegral}]
We consider the distribution of $T_n$ conditionally given
the observations used to construct the initial estimators
$\hat b_n$ and $\hat g_n$. By passing to subsequences of $n$, we may assume that these sequences converge
almost surely  to $b$ and $g$ relative to the uniform norm.
In the proof of distributional convergence the initial estimators
$\hat b_n$ and $\hat g_n$ may therefore be understood to be deterministic
sequences that converge to limits $b$ and $g$.

The estimator (\ref{EqRegressionEstimator}) is a sum
$T_n=T_n^{(1)}+T_n^{(2)}$ of a linear and quadratic part.
The (conditional) variance of the linear term $T_n^{(1)}$
is of the order $1/n$, which is of smaller order than
$k_n/n^2$. It follows that $(T_n^{(1)}-\E T_n^{(1)})/(\sqrt{k_n}/n)$ 
tends to zero in probability.

To study the quadratic part $T_n^{(2)}$ 
we apply Theorem~\ref{TheoremMainResult} with the kernel
$K_n$ of the theorem taken equal to the present $K_{k_n,\hat g_n}$
and the $Y_r$ of the theorem taken equal to the
present $Y_r-\hat b_n(X_r)$. For given functions $b_1$ and $g_1$, set
\begin{align*}
\m_q( b_1)(x)&= \E\bigl(|Y_1-b_1(X_1)|^q\given X_1=x\bigr)
= \E\bigl(|\e_1+(b-b_1)(x)|^q\given X_1=x\bigr),\\ 
k_n(b_1,g_1) 
&=\int\!\!\int (\m_2(b_1)\times \m_2(b_1))\,K_{k_n, g_1}^2\,d(G\times G).
\end{align*}
The function $\m_q(\hat b_n)$ converges
uniformly to the function $\m_q(b)$, which is uniformly
bounded by assumption, for $q=1$, $q=2$ and some $q>2$.
Furthermore $K_{k_n,\hat g_n}=K_{k_n,g}
\sqrt{g\times g/\hat g_n\times \hat g_n}$, where the 
function $g\times g/\hat g_n\times \hat g_n$
converges uniformly to one. Therefore, the conditions of
Theorem~\ref{TheoremMainResult} (for the case that the observations
are non-i.i.d.; cf. the remark following the theorem) are satisfied by
Theorem~\ref{TheoremWavelets} or~\ref{TheoremFourier}. 
Hence the sequence 
$(T_n^{(2)}-\E T_n^{(2)})/\sqrt{\hat k_n}/n^2$ tends to a standard normal
distribution, for $\hat k_n=k_n(\hat b_n,\hat g_n)$. From the
conditions on the initial estimators it follows that $\hat k_n/k_n(b,g)\ra 1$. 
Here $k_n(b,g)$ is of the order the dimension $k_n$ of the kernel.

Let $T_n(b_1,g_1)$ be as $T_n$, but with the initial estimators
$\hat b_n$ and $\hat g_n$ replaced by $b_1$ and $g_1$.
Its expectation   is given by
\begin{align*}
&e(b_1,g_1)=\E _{b,g}T_n(b_1,g_1)\\
&\qquad= \int b_1^2\,dG+\int 2b_1(b-b_1)\,dG
+\int\!\!\int (b-b_1)\times (b-b_1)\,K_{k_n,g_1}\,dG\times G.
\end{align*}
In particular $e(b,g)=\int b^2\,dG$. Using the fact that
$K_{k_n,g}$ is an orthogonal projection in $L_2(G)$ we can write
\begin{align}
e(b_1,g_1)-e(b,g)
&=-\int (b_1-b)^2\,dG
+\int\!\!\int (b-b_1)\times (b-b_1)\,K_{k_n,g_1}\,dG\times G\nonumber\\
&=-\bigl\|(I-K_{k_n,g})(b_1-b)\bigr\|_G^2
\label{EqEMinE}\\
&\qquad\qquad
+\int\!\!\int (b-b_1)\times (b-b_1)\,(K_{k_n,g_1}-K_{k_n,g})\,dG\times G.
\nonumber
\end{align}
By the definition of $K_{k_n,g}$ the absolute value of the first term
on the right can be bounded as
$$\Bigl\|(b-b_1)-\lin\Bigl({e_1\over \sqrt g},\ldots,{e_k\over \sqrt g}\Bigr)
\Bigr\|_G^2
=\Bigl\|(b-b_1)\sqrt g-\lin\bigl(e_1,\ldots,e_k\bigr)\Bigr\|_\LEB^2.$$
By assumption $b$ is $\b$-H\"older and $g$ 
is $\g$-H\"older for some $\g\ge \b$ and bounded away from zero. 
Then $b\sqrt g$ is
$\b$-H\"older and hence its uniform distance to $\lin(e_1,\ldots,e_k)$
is of the order $(1/k)^\b$. If the norm of $\hat b_n$ in $C^\b[0,1]$
is bounded, then we can apply the same argument to 
the functions $\hat b_n\sqrt g$, uniformly in $n$, and conclude that
the expression in the display with $\hat b_n$
instead of $b_1$ is bounded above by $O_P(1/k_n)^{2\b}$.
If the projection is on the Haar basis and
$\hat b_n$ is contained in $\lin (e_1,\ldots, e_{k_n})$, then
the approximation error can be seen to be of the same order,
from the fact that the product of two projections on the Haar basis is
itself a projection on this basis.

For $h=(\sqrt g-\sqrt g_1)/\sqrt{gg_1}$ we can write
$${1\over \sqrt {g_1(x_1)}\sqrt {g_1(x_2)}}
-{1\over \sqrt {g(x_1)}\sqrt {g(x_2)}}
=h(x_1)\Bigl({1\over \sqrt{g_1(x_2)}}\Bigr)
+h(x_2)\Bigl({1\over \sqrt{g(x_1)}}\Bigr).$$
If multiplied by a symmetric function in $(x_1,x_2)$ and
integrated with respect to $G\times G$, the arguments $x_1$ and $x_2$
in the second term can be exchanged.
The second term on the right in (\ref{EqEMinE}) can therefore be written
\begin{align*}
&\Bigl\langle K_{k_n,\l}\bigl((b-b_1)h\bigr),
(b-b_1)\Bigl({1\over \sqrt{g_1}}+{1\over \sqrt{g}}\Bigr) \Bigr\rangle_G\\
&\qquad\lesssim \bigl\|K_{k_n,\l}\bigl((b-b_1)h\bigr)\bigr\|_{G,3/2}
\|b-b_1\|_{G,3}\\
&\qquad\lesssim \bigl\|(b-b_1)h\bigr\|_{G,3/2}\|b-b_1\|_{G,3}
\lesssim \|b-b_1\|_{\l,3}\|h\|_{\l,3}\|b-b_1\|_{\l,3}.
\end{align*}
Here $\|\cdot\|_{G,3}$ is the $L_3(G)$-norm, we use the fact
that $L_2$-projection on a wavelet basis decreases 
$L_p$-norms for $p=3/2$ up to constants, and the multiplicative
constants depend on uniform upper and lower bounds on the functions 
$g_1$ and $g$. We evaluate this expression for
$b_1=\hat b_n$ and $g_1=\hat g_n$, and see that it is
of the order $O\bigl(\|\hat b_n-b\|_3^2\|\hat g_n-g\|_3\bigr)$.

Finally we note that $\hat{\E}_{b,g} T_n=e(\hat b_n,\hat g_n)$ and
combine the preceding bounds.
\end{proof}

\section{Appendix: proofs for Section~\ref{SectionProjectionKernels}}

\begin{lemma}
\label{LemmaNormProjection}
The kernel of an orthogonal projection on a $k$-dimensional space
has operator norm $\|K_k\|_2=1$,
and square $L_2(G\times G)$-norm 
$\int\int K_k^2\,d(G\times G)=k$. 
\end{lemma}

\begin{proof}
The operator norm is one, because an orthogonal projection decreases norm
and acts as the identity on its range.
It can be verified that the kernel of a kernel operator
is uniquely defined by the operator. Hence the kernel of a projection 
on a $k$-dimensional space can be written in the form 
(\ref{EqProjectionKernelOnBasis}), from which the $L_2$-norm can be computed. 
\end{proof}

\beginskip
\begin{proof}[Proof of Proposition~\ref{TheoremHaar}]
Because the density of $G$ and the function
$\m_2$ are bounded above and below the $L_2(G\times G)$-norm
$k_n$ as in (\ref{EqDefinitionkn}) is of the same order as the
dimension $k_n=2^I$ of the kernel. On the set
where it is nonzero, the kernel $K_k$ in (\ref{EqHaarKernel})
is of the order $k_n$. We choose the sets $\X_{n,m}$
to be unions of $l_n$ adjacent sets of the type
$(j2^{-I},(j+1)2^{-I}]=\bigl(j/k_n,(j+1)/k_n\bigr]$,
giving $M_n=k_n/l_n$ sets $\X_{n,m}$. Then 
$${1\over k_n}\int_{\X_{n,m}}\!\!\int_{\X_{n,m}} K_n^2\ 
(\m_2\times\m_2)\,d(G\times G)\lesssim {1\over k_n}
k_n^2\, l_n\Bigl({1\over k_n}\Bigr)^2={1\over M_n}.$$
It follows that (\ref{EqConditionTwo}) is  satisfied for any
$M_n\ra\infty$.

Because $M_n^{-1}\lesssim G(\X_{n,m})\lesssim M_n^{-1}$
conditions (\ref{EqConditionTwoHalf}) and (\ref{EqConditionThree})
are satisfied if $M_n\ra\infty$ with $M_n\le n$. Finally
(\ref{EqConditionFour}) can be replaced by (\ref{EqConditionFourSimplified})
in view of Lemma~\ref{LemmaConditionFourSimplified}, and this
is satisfied, for instance,
if we choose $M_n=n$.
\end{proof}

\endskip

\begin{proof}[Proof of Proposition~\ref{TheoremWavelets}]
We can reexpress the wavelet expansion (\ref{EqWaveletExpansionf}) 
to start from level $I$ as 
$$f(x)=\sum_{j\in\ZZ^d}\sum_{v\in\{0,1\}^d}
\langle f,\psi_{I,j}^v\rangle \psi_{I,j}^v(x)
+\sum_{i=I+1}^\infty\sum_{j\in\ZZ^d}\sum_{v\in\{0,1\}^d-\{0\}}
\langle f,\psi_{i,j}^v\rangle \psi_{i,j}^v(x).$$
The projection kernel $K_k$  sets the coefficients
in the second sum equal to zero, and hence can also be
expressed as 
$$K_k(x_1,x_2)
=\sum_{j\in J_I}\sum_{v\in\{0,1\}^d}
\psi_{I,j}^v(x_1)\psi_{I,j}^v(x_2).$$
The double integral of the square of this function
over $\RR^{2d}$ is equal to the number of terms in the double
sum (cf.\ (\ref{EqProjectionKernelOnBasis}) and the remarks following it),
which is $O(2^{Id})$.
The support of only a small fraction of functions in the
double sum intersects the boundary of $\X$.
Because also the density of $G$ and the function $\m_2$
are bounded above and below, it follows that the 
weighted double integral $k_n$ of $K_k^2$ relative to $G$
as in (\ref{EqDefinitionkn}) is also of the exact order 
$O(2^{Id})$.

Each function 
$(x_1,x_2)\mapsto \psi_{I,j}^v(x_1)\psi_{I,j}^v(x_2)$ 
has uniform norm bounded above by $2^{Id}$ times the 
uniform norm of the base wavelet of which it is a shift and dilation.
A given point $(x_1,x_2)$ belongs to the support of fewer than $C_1^d$
of these functions, for a constant $C_1$ that depends on
the shape of the support of the wavelets. Therefore, the uniform
norm of the kernel $K_k$ is of the order $k_n$.

By assumption each function $\psi_{I,j}^v$ is supported within a set
of the form $2^{-I}(C+j)$ for a given cube $C$ that depends on the
type of wavelet, for any $v$. It follows that the function
$(x_1,x_2)\mapsto \psi_{I,j}^v(x_1)\psi_{I,j}^v(x_2)$ vanishes outside
the cube $2^{-I}(C+j)\times 2^{-I}(C+j)$.  There are $O(2^{Id})$ of
these cubes that intersect $\X\times\X$; these intersect the diagonal
of $\X\times\X$, but may be overlapping.  We choose the sets
$\X_{n,m}$ to be blocks (cubes) of $l_n^d$ adjacent cubes
$2^{-I}(C+j)$, giving $M_n=O(k_n/l_n^d)$ sets $\X_{n,m}$.  [In the
case $d=1$, the ``cubes'' are intervals and they can be ordered linearly;
the meaning of ``adjacent'' is then clear. For $d>1$ cubes are
``adjacent'' in $d$ directions.  We stack $l_n$ cubes $2^{-I}(C+j)$ in
each direction, giving cubes $\X_{n,m}$ of sides with lengths $l_n$
times the length of a cube $2^{-I}(C+j)$.]

Because the kernels are bounded by a multiple of $k_n$,
condition (\ref{EqConditionFour}) is implied by 
(\ref{EqConditionFourSimplified}), in view of Lemma~\ref{LemmaConditionFourSimplified},
The latter condition reduces to $M_n^{-1}k/n\ra0$, the probabilities
$G(\X_{n,m})$ being of the order $1/M_n$.

The set of cubes  $2^{-I}(C+j)$ that intersects more than one
set $\X_{n,m}$ is of the order $M_n^{1/d} k_n^{1-1/d}$.
To see this picture the set $\X$ as a supercube
consisting of the $M$ cubes $\X_{n,m}$, stacked together in a
$M^{1/d}\times\cdots\times M^{1/d}$-pattern. For each coordinate
$i=1,\ldots,d$ the stack of cubes $\X$ can be sliced in $M^{1/d}$
layers each consisting of $(M^{1/d})^{d-1}$ cubes $\X_{m,n}$,
which are $l_n(k_n^{1/d})^{d-1}=l_n^d(M_n^{1/d})^{d-1}$ cubes $2^{-I}(C+j)$.
The union of the boundaries of all slices ($i=1,\ldots,d$ and
$M_n^{1/d}$ slices for each $i$) contains the union
of the boundaries of the sets $\X_{n,m}$. 
The boundary between two particular slices is intersected by
at most $C_2(k_n^{1/d})^{d-1}$ cubes $2^{-I}(C+j)$, for a constant
$C_2$ depending on the amount of overlap between the cubes.
Thus in total of the order $d M_n^{1/d} (k_n^{1/d})^{d-1}$ cubes
intersect some boundary.

If $K_k(x_1,x_2)\not=0$, then there exists $j$ and $v$ with
$\psi_{I,j}^v(x_1)\psi_{I,j}^v(x_2)\not=0$, which implies
that there exists $j$ such that $x_1,x_2\in 2^{-I}(C+j)$. If 
the cube $2^{-I}(C+j)$ is contained in some $\X_{n,m}$, then
$(x_1,x_2)\in\X_{n,m}\times\X_{n,m}$. In the other
case $2^{-I}(C+j)$ intersects the boundary of some $\X_{n,m}$.
It follows that the set of $(x_1,x_2)$ in the complement of 
$\cup_m\X_{n,m}\times\X_{n,m}$ where $K_k(x_1,x_1)\not=0$
is contained in the union $U$  of all cubes $2^{-I}(C+j)$
that intersect the boundary of some $\X_{n,m}$. The integral of $K_k^2$
over this set satisfies
\begin{align*}
{1\over k_n}\dint_U K_k^2\,d(G\times G)
&\lesssim {1\over k_n}\,k_n^2\, (G\times G)(U)\\
&\lesssim
 {1\over k_n}\,k_n^2\, M_n^{1/d} k_n^{1-1/d}  \,\Bigl({1\over k_n}\Bigr)^2
=\Bigl({M_n\over k_n}\Bigr)^{1/d}.
\end{align*}
Here we use that $G\bigl(2^{-I}(C+j)\bigr)\lesssim 1/k_n$.
This completes the verification of (\ref{EqConditionOne}).

By the spatial homogeneity of the wavelet basis, the contributions
of the sets $\X_{n,m}\times\X_{n,m}$ to the integral of $K_k^2$
are comparable in magnitude. Hence condition (\ref{EqConditionTwo}) 
is  satisfied for any $M_n\ra\infty$. 

In order to satisfy conditions (\ref{EqConditionTwoHalf}) and
(\ref{EqConditionThree}) we must choose $M_n\ra\infty$ with $M_n\lesssim n$.
This is compatible with choices such that
$M_n/k_n\ra 0$ and $M_n^{-1}k/n\ra0$.
\end{proof}

\begin{figure}
\label{FigurePartitionWithOverlap}
\centerline{\resizebox{1.6in}{!}{\includegraphics{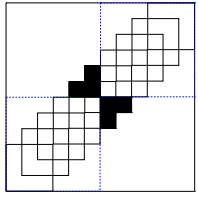}}}
\caption{The support cubes of the wavelets and the
bigger cubes $\X_{n,m}\times\X_{n,m}$.}
\end{figure}

\begin{proof}[Proof of Proposition~\ref{TheoremFourier}]
Because $K_k$ is an orthogonal projection on a $(2k+1)$-dimensional
space, Lemma~\ref{LemmaNormProjection} gives that the operator
norm satisfies $\|K_k\|=1$ and that the numbers
$k_n$ as in (\ref{EqDefinitionkn}) but with $\m_2=1$ are equal to 
$\int\int K_k^2\,d\LEB\,d\LEB=2k+1$.

By the change of variables $x_1-x_2=u$, $x_1+x_2=v$ we find,
for any $\e\in (0,\pi]$, and $K_k(x_1,x_2)=D_k(x_1-x_2)$, 
\begin{align*}
\int_{-\pi}^\pi\!\!\int_{-\pi}^\pi1_{|x_1-x_2|>\e}
K_k^2(x_1,x_2)\,dx_1\,dx_2
&=2\int_\e^{2\pi}\!\!\int_{u-2\pi}^{2\pi-u} D_k^2(u)\thalf\, dv\,du\\
&=2\int_\e^{2\pi} D_k^2(u)(2\pi-u)\, du.
\end{align*}
By the symmetry of the Dirichlet kernel about $\pi$ we can rewrite
$\int_\pi^{2\pi} D_k^2(u)(2\pi-u)\,du$ as $\int_0^\pi D_k^2(u)u\,du$.
Splitting the integral on the right side of the preceding display
over the intervals $(\e,\pi]$ and $(\pi,2\pi]$, and rewriting
the second integral, we see that the preceding display is equal to
$$2\int_\e^\pi D_k^2(u)(2\pi-u)\,du+2\int_0^\pi D_k^2(u)u\,du
=4\pi\int_\e^\pi D_k^2(u)\,du+2\int_0^\e D_k^2(u)u\,du.$$
For $\e=0$ this expression is equal to the square $L_2$-norm of the
kernel $K_k$, which shows that $4\pi\int_0^\pi D_k^2(u)\,du=2k+1$.
On the interval $(\e,\pi)$ the kernel $D_k$ is bounded above by
$\bigl(2\pi \sin(\thalf\e)\bigr)^{-1}$. Therefore, the
preceding display is bounded above by
$${4\pi\over \bigl(2\pi \sin(\thalf\e)\bigr)^2}
\int_\e^\pi \,du+2\e\int_0^\e D_k^2(u)\,du
\lesssim {1\over \sin^2 (\thalf\e)}+\e k.$$
We conclude that, for small $\e>0$,
$${1\over 2k+1}\int_{-\pi}^\pi\!\!\int_{-\pi}^\pi1_{|x_1-x_2|>\e}
K_k^2(x_1,x_2)\,dx_1\,dx_2\lesssim \e+{1\over \e^2 k}.$$
This tends to zero as $k\ra\infty$
whenever $\e=\e_k\da 0$ such that $\e\gg 1/\sqrt k$.

We choose a partition $(-\pi,\pi]=\cup_m\X_{n,m}$ in $M_n=2\pi/\d$
intervals of length $\d$ for $\d\ra0$ with $\d\gg \e$ and
$\e$ satisfying the conditions of the preceding paragraph. Then 
the complement of $\cup_m\X_{n,m}\times\X_{n,m}$ is contained
in $\{(x_1,x_2): |x_1-x_2|>\e\}$ except for a set of
$2(M_n-1)$ triangles, as indicated in Figure~\ref{FigurePartitionWithStrip}.
In order to verify (\ref{EqConditionOne}) it suffices to show that
$(2k+1)^{-1}$ times the integral of $K_k^2$ over the union of the triangles
is negligible. Each triangle has sides of length of the order $\e$,
whence, for a typical triangle $\Delta$, by the change of variables
$x_1-x_2=u$, $x_2=v$, and an interval $I$ of length of the order $\e$,
$$\dint_\Delta K_k^2(x_1,x_2)\,dx_1\,dx_2
\lesssim \int_I\!\int_0^\e D_k^2(u)\,du\,dv
\lesssim  \e (2k+1).$$
Hence (\ref{EqConditionOne}) is satisfied if
$2(M_n-1)\e\ra 0$, i.e.\ $\e\ll\d$.

Because $\int_{\X_{n,m}}\int_{X_{n,m}} K_k^2\,d(\LEB\times \LEB)$ is
independent of $m$, (\ref{EqConditionTwo}) is satisfied as soon as
the number of sets in the partitions tends to infinity.

Because $0\le K_k\le 2k+1$, condition (\ref{EqConditionFour}) is 
implied by (\ref{EqConditionFourSimplified}), which is satisfied
if $\d\ll n/k$.

The desired choices $1/\sqrt k\ll \e\ll\d\ll n/k$ are compatible,
as by assumption $k/n^2\ra 0$.
\end{proof}

\begin{figure}
\label{FigurePartitionWithStrip}
\centerline{\resizebox{1.6in}{!}{\includegraphics{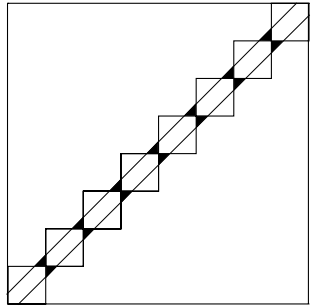}}}
\caption{The triangles used in the proofs of Theorems~\ref{TheoremFourier}
and~\ref{TheoremConvolution}, and the sets $\X_{n,m}\times\X_{n,m}$.}
\end{figure}

\begin{proof}[Proof of Proposition~\ref{TheoremConvolution}]
Without loss of generality we can assume that $\int|\phi|\,d\LEB=1$.
By a change of variables
$$\int K_\s^2\,d(G\times G)
={1\over \s}\int \phi^2(v)\int g(x-\s v)g(x)\,dx\,dv.$$
Here
$\bigl|\int g(x-\s v)g(x)\,dx\bigr|\le \|g\|_\infty$ and,
as $\s\da 0$,
$$\Bigl|\int g(x-\s v)g(x)\,dx-\int g^2(x)\,dx\Bigr|
\le \|g\|_\infty
\int \bigl|g(x-\s v)-g(x)\bigr|\,dx\ra 0,$$
for every fixed $v$, by the $L_1$-continuity theorem.
We conclude by the dominated convergence theorem that 
$\s\int K_\s^2\,d(G\times G)\ra \int g^2\,d\LEB\int\phi^2\,d\LEB$.
Because $\m_2$ is bounded away from 0 and $\infty$,
the numbers $k_n$ defined in (\ref{EqDefinitionkn}) are of the
exact order $\s^{-1}$.

By another change of variables, followed by an application of
the Cauchy-Schwarz inequality, for any $f\in L_2(G)$,
\begin{align*}
\int (K_\s f)^2\,dG&=\int\Bigl(\int\phi(v)(fg)(x-\s v)\,dv\Bigr)^2\,dG(x)\\
&\le \|g\|_\infty^2 \int\!\!\int|\phi|(v)(f^2g)(x-\s v)\,dv\,dx
= \|g\|_\infty^2 \int f^2g\,d\LEB.
\end{align*}
Therefore, the operator norms of the operators $K_\s$
are uniformly bounded in $\s> 0$.

We choose a partition $\RR=\cup_m\X_{n,m}$ consisting of
two infinite intervals $(-\infty,-a]$ and $(a,\infty)$
and a regular partition of the interval $(-a,a]$ in such a
way that every partitioning set satisfies $G(\X_{n,m})\le \d$.
We can achieve this with a partition in $M_n=O(1/\d)$
sets.

Because $|K_\s|$ is bounded by a multiple of $\s^{-1}$,
condition (\ref{EqConditionFour}) is implied
by (\ref{EqConditionFourSimplified}), which takes
the form $\d/(\s n)\ra 0$, in view of 
Lemma~\ref{LemmaConditionFourSimplified}.

For an arbitrary partitioning set $\X_{n,m}$, 
\begin{align*}
\s\int_{\X_{n,m}}\!\!\int_{\X_{n,m}}K_\s^2\,d(G\times G)
&\le \int_{\X_{n,m}}\int \phi^2(v)\,  g(x-\s v)g(x)\,dv\,dx\\
&\le \|g\|_\infty\int \phi^2(v)\,dv\, G(\X_{n,m}).
\end{align*}
It follows that (\ref{EqConditionTwo}) is satisfied
as soon as $\d\ra0$.

Finally, we verify condition (\ref{EqConditionOne}) in two steps.
First, for any $\e\da0$,  by the change of variables
$x_1-x_2=v$, $x_2=x$,
\begin{align*}
\s\!\!\dint_{|x_1-x_2|>\e} K_\s^2\,d(G\times G)
&=\dint_{|v|>\e/\s}\phi^2(v)\,g(x-\s v)g(x)\,dx\,dv\\
&\le \|g\|_\infty\int_{|v|>\e/\s}\phi^2(v)\,dv.
\end{align*}
This converges to zero as $\s\ra0$ for any $\e=\e_\s>0$ with $\e\gg\s$.
Second, for $\e\ll\d$ the complement of the set
$\cup_m\X_{n,m}\times\X_{n,m}$ is contained
in $\{(x_1,x_2): |x_1-x_2|>\e\}$ except for a set of
$2(M_n-1)$ triangles, as indicated in Figure~\ref{FigurePartitionWithStrip}.
In order to verify (\ref{EqConditionOne}) it suffices to show that
$\s$ times the integral of $K_\s^2$ over the union of the triangles
is negligible. Each triangle has sides of length of the order $\e$,
whence, for a typical triangle $\Delta$, with projection $I$
on the $x_1$-axis,
$$\s\dint_\Delta K_\s^2\,d(G\times G)
\lesssim \int_I\!\!\int_{|v|<\e/\s} \phi^2(v)\,g(x-\s v)g(x)\,dv\,dx
\le  \e \|g\|_\infty \int \phi^2(v)\,dv.$$
The total contribution of all triangles is $2(M_n-1)$
times this expression. Hence (\ref{EqConditionOne}) is satisfied if
$2(M_n-1)\e\ra 0$, i.e.\ $\e\ll\d$.

The preceding requirements can be summarized
as $\s\ll\e\ll\d\ll\s n$, and are compatible.
\end{proof}


\begin{proof}[Proof of Proposition~\ref{TheoremSplines}]
Inequality (\ref{EqCj}) implies that $c_j(f)=0$ for every $f$ 
that vanishes outside the interval $(t_j',t_{j+r}')$,
whence the representing function $g_j$ is supported
on this interval. It follows that the function 
$(x_1,x_2)\mapsto N_j(x_1)c_j(x_2)$
vanishes outside the square $[t_j',t_{j+r}']\times[t_j',t_{j+r}']$,
which has area of the order $l^{-2}$. 
We form a partition $(0,1]=\cup_m\X_{n,m}$ by selecting
subsets $0=s_0^l<s_1^l<\cdots<s_{M_n}^l=1$ of the basic knot sequences
such that $M_n^{-1}\lesssim s_{i+1}^l-s_i^l\lesssim M_n^{-1}$ for every $i$
and define $\X_{n,m}=(s_{m-1}^l,s_{m}^l]$.
The numbers $M_n$ are chosen integers much smaller than $l_n$,
and we may set $s_i^l=t_{ip}^l$ for $p=\lfloor l_n/M_n\rfloor$.

Because $K_k$ is a projection on $S_r(T,d)$
and the function $x_1\mapsto K_k(x_1,x_2)$ is contained 
in $S_r(T,d)$ for every $x_2$,
it follows that $\int K_k(x_1,x_2)K_k(x_1,x_2)\,dx_1=K_k(x_2,x_2)$ for
every $x_2$, and hence
\begin{align*}
&\int\!\!\int K_k(x_1,x_2)^2\,dx_1\,dx_2
=\int K_k(x_1,x_1)\,d\LEB(x_1)\\
&\qqqquad=\int \sum_j N_j(x_1)c_j(x_1)\,dx_1
=\sum_jc_j(N_j)=\sum_j 1=k,
\end{align*}
because the identities $N_i=K_kN_i=\sum_j c_j(N_i)N_j$ imply that
$c_j(N_i)=\d_{ij}$ by the linear independence of the B-splines.
Because the density of $G$ and the function
$\m_2$ are bounded above and below the $L_2(G\times G)$-norm
$k_n$ as in (\ref{EqDefinitionkn}) is of the same order as the
dimension $k_n=r+l_nd$ of the spline space. 

Inequality (\ref{EqCj}) implies that the norm of the linear map $c_j$,
which is the infinity norm $\|c_j\|_\infty$ of the representing
function, is bounded above by a constant times $(t_{j+r}'-t_j')^{-1}$,
which is of the order $k_n$.  Therefore,
\begin{align*}
&{1\over k_n}\int\!\!\int_{(\cup_m\X_{n,m}\times\X_{n,m})^c} K_n^2\ 
(\m_2\times\m_2)\,d(G\times G)\\
&\qquad\lesssim {1\over k_n} \, k_n^2\,\|\m_2\|_\infty^2\,
\LEB\Bigl(\bigcup_j(t_j',t_{j+r}']\times (t_j',t_{j+r}']
-\bigcup_m (s_{m-1},s_m]\times (s_{m-1},s_m]\Bigr).
\end{align*}
The set in the right side is the union of $M_n$ cubes
of areas not bigger than the area of the sets 
$(t_j',t_{j+r}']\times (t_j',t_{j+r}']$, which is bounded
above by a constant times $k_n^{-2}$. 
(See Figure~\ref{FigurePartitionWithOverlap}.) The preceding display
is therefore bounded above by
$${1\over k_n}k_n^2\,\|\m_2\|_\infty^2\, M_n\,{1\over k_n^2}.$$
For $M_n/k_n\ra0$ this tends to zero. This completes
the verification of (\ref{EqConditionOne}).

The verification of the other conditions follows the same
lines as in the case of the wavelet basis.
\end{proof}


\bibliographystyle{elsarticle-num} 

\bibliography{normU}

\begin{thebibliography}{10}
\expandafter\ifx\csname url\endcsname\relax
  \def\url#1{\texttt{#1}}\fi
\expandafter\ifx\csname urlprefix\endcsname\relax\def\urlprefix{URL }\fi
\expandafter\ifx\csname href\endcsname\relax
  \def\href#1#2{#2} \def\path#1{#1}\fi

\bibitem{FreedmanFestschrift}
J.~Robins, L.~Li, E.~Tchetgen, A.~van~der Vaart,
  \href{http://dx.doi.org/10.1214/193940307000000527}{Higher order influence
  functions and minimax estimation of nonlinear functionals}, in: Probability
  and statistics: essays in honor of {D}avid {A}. {F}reedman, Vol.~2 of Inst.
  Math. Stat. Collect., Inst. Math. Statist., Beachwood, OH, 2008, pp.
  335--421.
\newblock \href {http://dx.doi.org/10.1214/193940307000000527}
  {\path{doi:10.1214/193940307000000527}}.
\newline\urlprefix\url{http://dx.doi.org/10.1214/193940307000000527}

\bibitem{BickelRitov88}
P.~J. Bickel, Y.~Ritov, Estimating integrated squared density derivatives:
  sharp best order of convergence estimates, Sankhy\=a Ser. A 50~(3) (1988)
  381--393.

\bibitem{BirgeMassart95}
L.~Birg{\'e}, P.~Massart, Estimation of integral functionals of a density, Ann.
  Statist. 23~(1) (1995) 11--29.

\bibitem{Laurent96}
B.~Laurent, Efficient estimation of integral functionals of a density, Ann.
  Statist. 24~(2) (1996) 659--681.

\bibitem{Laurent97}
B.~Laurent, Estimation of integral functionals of a density and its
  derivatives, Bernoulli 3~(2) (1997) 181--211.

\bibitem{LaurentMassart00}
B.~Laurent, P.~Massart, Adaptive estimation of a quadratic functional by model
  selection, Ann. Statist. 28~(5) (2000) 1302--1338.

\bibitem{RobinsMetrika}
J.~Robins, L.~Li, E.~Tchetgen, A.~W. van~der Vaart,
  \href{http://dx.doi.org/10.1007/s00184-008-0214-3}{Quadratic semiparametric
  von {M}ises calculus}, Metrika 69~(2-3) (2009) 227--247.
\newblock \href {http://dx.doi.org/10.1007/s00184-008-0214-3}
  {\path{doi:10.1007/s00184-008-0214-3}}.
\newline\urlprefix\url{http://dx.doi.org/10.1007/s00184-008-0214-3}

\bibitem{vdVStatSci}
A.~van~der Vaart, Higher order tangent spaces and influence functions, Statist.
  Sci. 29~(4) (2014) 679--686.
\newblock \href {http://dx.doi.org/10.1214/14-STS478}
  {\path{doi:10.1214/14-STS478}}.

\bibitem{RobinsAOS}
E.~Tchetgen, L.~Li, J.~Robins, A.~van~der Vaart, Higher order estimating
  equations for high-dimensional semiparametric models, preprint.

\bibitem{RobinsvdV06}
J.~Robins, A.~van~der Vaart, Adaptive nonparametric confidence sets, Ann.
  Statist. 34~(1) (2006) 229--253.

\bibitem{Weber}
N.~C. Weber, \href{http://dx.doi.org/10.1017/S0305004100061168}{Central limit
  theorems for a class of symmetric statistics}, Math. Proc. Cambridge Philos.
  Soc. 94~(2) (1983) 307--313.
\newblock \href {http://dx.doi.org/10.1017/S0305004100061168}
  {\path{doi:10.1017/S0305004100061168}}.
\newline\urlprefix\url{http://dx.doi.org/10.1017/S0305004100061168}

\bibitem{BhattacharyaGhosh}
R.~N. Bhattacharya, J.~K. Ghosh,
  \href{http://dx.doi.org/10.1016/0047-259X(92)90038-H}{A class of
  {$U$}-statistics and asymptotic normality of the number of {$k$}-clusters},
  J. Multivariate Anal. 43~(2) (1992) 300--330.
\newblock \href {http://dx.doi.org/10.1016/0047-259X(92)90038-H}
  {\path{doi:10.1016/0047-259X(92)90038-H}}.
\newline\urlprefix\url{http://dx.doi.org/10.1016/0047-259X(92)90038-H}

\bibitem{JammaladakaRaoJanson}
S.~R. Jammalamadaka, S.~Janson, Limit theorems for a triangular scheme of
  {$U$}-statistics with applications to inter-point distances, Ann. Probab.
  14~(4) (1986) 1347--1358.

\bibitem{deJong87}
P.~de~Jong, \href{http://dx.doi.org/10.1007/BF00354037}{A central limit theorem
  for generalized quadratic forms}, Probab. Theory Related Fields 75~(2) (1987)
  261--277.
\newblock \href {http://dx.doi.org/10.1007/BF00354037}
  {\path{doi:10.1007/BF00354037}}.
\newline\urlprefix\url{http://dx.doi.org/10.1007/BF00354037}

\bibitem{deJong90}
P.~de~Jong, \href{http://dx.doi.org/10.1016/0047-259X(90)90040-O}{A central
  limit theorem for generalized multilinear forms}, J. Multivariate Anal.
  34~(2) (1990) 275--289.
\newblock \href {http://dx.doi.org/10.1016/0047-259X(90)90040-O}
  {\path{doi:10.1016/0047-259X(90)90040-O}}.
\newline\urlprefix\url{http://dx.doi.org/10.1016/0047-259X(90)90040-O}

\bibitem{Mikosch}
T.~Mikosch, A weak invariance principle for weighted {$U$}-statistics with
  varying kernels, J. Multivariate Anal. 47~(1) (1993) 82--102.

\bibitem{Morris}
C.~Morris, Central limit theorems for multinomial sums, Ann. Statist. 3 (1975)
  165--188.

\bibitem{Ermakov}
M.~S. Ermakov, Asymptotic minimaxity of chi-squared tests, Teor. Veroyatnost. i
  Primenen. 42~(4) (1997) 668--695.

\bibitem{KerkPicard96}
G.~Kerkyacharian, D.~Picard, Estimating nonquadratic functionals of a density
  using {H}aar wavelets, Ann. Statist. 24~(2) (1996) 485--507.

\bibitem{RobinsMinimax}
J.~Robins, E.~Tchetgen~Tchetgen, L.~Li, A.~van~der Vaart,
  \href{http://dx.doi.org/10.1214/09-EJS479}{Semiparametric minimax rates},
  Electron. J. Stat. 3 (2009) 1305--1321.
\newblock \href {http://dx.doi.org/10.1214/09-EJS479}
  {\path{doi:10.1214/09-EJS479}}.
\newline\urlprefix\url{http://dx.doi.org/10.1214/09-EJS479}

\bibitem{Neweyetal}
W.~K. Newey, F.~Hsieh, J.~M. Robins, Twicing kernels and a small bias property
  of semiparametric estimators, Econometrica 72~(3) (2004) 947--962.

\bibitem{Daubechies}
I.~Daubechies, Ten lectures on wavelets, Vol.~61 of CBMS-NSF Regional
  Conference Series in Applied Mathematics, Society for Industrial and Applied
  Mathematics (SIAM), Philadelphia, PA, 1992.

\bibitem{DevoreLorentz}
R.~A. DeVore, G.~G. Lorentz, Constructive approximation, Vol. 303 of
  Grundlehren der Mathematischen Wissenschaften [Fundamental Principles of
  Mathematical Sciences], Springer-Verlag, Berlin, 1993.

\bibitem{vdVAS}
A.~W. van~der Vaart, Asymptotic statistics, Vol.~3 of Cambridge Series in
  Statistical and Probabilistic Mathematics, Cambridge University Press,
  Cambridge, 1998.

\bibitem{GineLatalaZinn}
E.~Gin{\'e}, R.~Latala, J.~Zinn, Exponential and moment inequalities for
  {$U$}-statistics, in: High dimensional probability, II (Seattle, WA, 1999),
  Vol.~47 of Progr. Probab., Birkh\"auser Boston, Boston, MA, 2000, pp. 13--38.

\end{thebibliography}

\end{document}